\documentclass[letterpaper]{article} 
\usepackage{arxiv}  
\usepackage{times}  
\usepackage{helvet} 
\usepackage{courier}  
\usepackage[hyphens]{url}  
\usepackage{graphicx} 
\usepackage{natbib}  
\usepackage{caption} 
\usepackage[utf8]{inputenc} 
\usepackage[T1]{fontenc}    
\usepackage{hyperref}       
\usepackage{url}            
\usepackage{booktabs}       
\usepackage{amsfonts}       
\usepackage{nicefrac}       
\usepackage{microtype}      
\usepackage{graphicx}
\usepackage{subfigure}
\usepackage{subfiles}
\usepackage{booktabs} 
\usepackage{algorithm}
\usepackage{algorithmic}
\usepackage{amsmath,amsthm,amssymb}
\usepackage[utf8]{inputenc}
\usepackage[english]{babel}
\usepackage{tikz}
\usetikzlibrary{positioning}
\usepackage{pgflibraryshapes}
\usepackage{hyperref}
\usepackage{xr}

\usepackage{mathtools}

\newtheorem{definition}{Definition}
 
\newtheorem{theorem}{Theorem}

\newtheorem{lemma}[theorem]{Lemma}
\newtheorem{remark}{Remark}

\DeclarePairedDelimiter\abs{\lvert}{\rvert}%
\makeatletter
\let\oldabs\abs
\def\abs{\@ifstar{\oldabs}{\oldabs*}}

\title{Identifying Best Fair Intervention}
 
\author{
  Ruijiang Gao \\
  University of Texas at Austin \\
  \texttt{ruijiang@utexas.edu} \\
   \And
  Han Feng \\
  University of California, Berkeley \\
  \texttt{han\_feng@berkeley.edu} \\
}

\begin{document}

\maketitle

\begin{abstract}
We study the problem of best arm identification with a fairness constraint in a given causal model. The goal is to find a soft intervention on a given node to maximize the outcome while meeting a fairness constraint by counterfactual estimation with only \textbf{\textit{partial}} knowledge of the causal model. The problem is motivated by ensuring fairness on an online marketplace. We provide theoretical guarantees on the  probability of error and empirically examine the effectiveness of our algorithm with a two-stage baseline.
\end{abstract}

\section{Introduction}
Automated decision-making is becoming more involved in people's lives, with applications in hiring ~\citep{policing2019,hire2019}, last mile delivery~\citep{liuOnTimeLastMile2018}, and online advertising~\citep{bottou2013counterfactual}. These applications have given rise to concerns about the transparency and the fairness of the decision-making process. It is vital to distinguish unfair algorithms from fair ones for better quality of service and substantial societal benefits like diversity, stability and safety. 

Efforts have been made in the design of algorithmic fairness in various settings~\citep{calders2010three,dwork2012fairness,bolukbasi2016man,nabi2018fair,joseph2016fairness}. Many of them define a quantitative measure of fairness and apply it on a dataset. The subjectivity of the definition of fairness has caused variation on different datasets and problems. ~\citet{kusner2017counterfactual} proposed a causal approach to the definition of fairness named \emph{counterfactual fairness}, which attempts to capture the variation caused by a counterfactually different sensitive attribute of the subject.

Our problem is motivated by the example of online advertisement in ~\citet{bottou2013counterfactual,sen2017identifying}, where there are many candidate algorithms for predicting click-through rate and the final revenue is determined by a complex causal relationship represented by a causal model. Causal graphs are useful representations of the complex causal relationship between the variables in many problems and systems~\citep{pearl2009causality}. \citet{sen2017identifying} considers the problem of finding the best algorithm that leads to the largest revenue by best-arm identification under budget constraints. Inspired by these works and recent advances on fairness accountability and transparency, we study the problem of identifying the best \emph{fair} algorithm that leads to maximum revenue while satisfying the fairness constraint under a fixed budget. 

Fair treatment to different customer groups is important to keep business alive. The deployment of an unfair facial recognition software can be detrimental to company's reputation and lead to legal risks~\citep{amz2019}. It is important for companies to prevent different treatment of customers based on their race, gender or other sensitive attributes. In light of this, we propose a novel group level fairness constraint on a continuous outcome, in the case the outcome is the revenue, the constraint bounds the variation of the revenue due to counterfactual change of sensitive attributes. After spending a limited budget, we try to minimize the probability of misidentifying the best \emph{fair} arm. Furthermore, our problem can be formalized as a best-arm identification problem with constraint, or top feasible arm identification in ~\citep{katz2019top}. This constraint is defined as counterfactual fairness in our paper and specifically studied in the causal bandits setting.

To summarize our contributions, we study the problem of identifying best fair intervention in a given causal model. Given \textbf{\textit{partial}} knowledge of the causal graph and the freedom of intervening in actions at certain nodes, we aim to choose a best sequence of interventions to discover which action leads to the largest expected value of the target node while meeting the fairness constraint. We provide theoretical guarantee of exponential decay rate on the probability of error of the proposed algorithm when best fair arm exists, and when it does not, show that our algorithm detects its non-existence with high probability asymptotically to 1. We empirically show its effectiveness against a naive two-stage baseline algorithm that spends half the budget for fair arm identification and the rest for best arm identification. The settings of the experiments will be discussed in detail in Section \ref{sec:na}.

\section{Related Work}
In the context of causal bandits, our work is closely related to ~\citep{lattimore2016causal, sen2017identifying,yabe2018causal}. ~\citet{lattimore2016causal} studied
the problem of identifying the best hard interventions on multiple variables, provided that the 
distribution of the parents of the target node is known under those interventions. In contrast, ~\citet{sen2017identifying} considered soft interventions that affect the mechanism between a `source’ node and its parents, far away from the target. Our work follows this setting and builds on it with a fairness constraint. ~\citep{yabe2018causal} considers an arbitrary set of binary interventions in a causal graph consisted of binary variables. Unlike previous work, we further impose a fairness constraint on arms. The goal is to find the best arm among all the fair arms, which makes the problem more challenging. We also consider both hard intervention and soft intervention with a different choice of directed cut in the causal graph to study the counterfactual fairness. 

The issue of fairness is receiving ever-increasing attention in the machine learning community. There are many research works on deriving fairness metrics or inducing fair algorithms from observational data ~\citep{nabi2018learning,nabi2018fair,gillen2018online,zhang2018fairness,jabbari2017fairness,dwork2012fairness}. Unfairness in algorithms can be caused by bias in data collection and by algorithm design~\citep{olteanu2019social}. It can be detrimental to the companies' business due to the loss of diversity of customers and the vulnerability to varied business environment~\citep{holstein2019improving,zhang2019group}. Thus, it is imperative to design an online algorithm that identifies fair algorithms for deployment beyond common business metrics. In this paper, we propose a counterfactual fairness criterion in the spirit of ~\citet{kusner2017counterfactual} on a continuous outcome. Our setting captures the case of revenue maximization on an online platform. 

In the study of fair bandits.
~\citet{patil2019achieving,li2019combinatorial,chen2019fair} considered a fairness problem to ensure the number of pulls of each arm is at least above a given fraction in settings like resource allocation. ~\citet{metevier2019offline} proposed a framework to use user-specified fairness measure with offline contextual bandit. Unlike this work, we focus on online causal bandits and propose a novel counterfactual fairness measure. 
\citet{joseph2016fairness} defined fairness as a worse arm never being preferred over a better one, which is important in applications like resource allocation. ~\citet{liu2017calibrated} adopted a close fairness definition where arms with similar quality distribution should be selected with similar probability. These definitions are closely related to individual fairness ~\citep{dwork2012fairness}. ~\citet{gillen2018online} considered a linear contextual bandit setting with an unknown fairness metric that is returned by the environment, but in our problem this feedback is not available. Unlike these works, we define the fairness as the expected difference between the revenue (outcome) and the counterfactual revenue for a different sensitive attribute in a given causal model. We deem an algorithm unfair if it has a clear advantage when varying \textit{only} the sensitive attribute in the causal model~\citep{pearl2009causality}. We also extend the counterfactual fairness in ~\citet{kusner2017counterfactual} to a group level fairness measure that addresses continuous variables. It is worth mentioning that the fairness in \citet{joseph2016fairness} is consistent with the objective in maximizing the total reward in traditional bandit literature, while our fairness criterion makes the arm with optimal reward possible to be unfair. ~\citet{katz2018feasible,katz2019top,katz2019feasible} has a similar problem setting as ours. The goal is to find a collection of arms whose means are feasible (inside a polyhedron) and maximize some linear reward function. Our work can also be viewed as an extension of feasible arm identification problem that considers the information leakage between arms due to the causal graph structure. Furthermore, we propose a novel fairness constraint base on the causal structure and estimate it by importance sampling. 


\section{Problem Statement}
In causal bandit problems~\citep{lattimore2016causal,sen2017identifying,yabe2018causal}, agent obtains different rewards for repeated intervention on a given causal model ~\citep{pearl2009causality}. A causal model is defined by a directed acyclic graph $\mathcal{G}$ over a set of random variables $\mathcal{V} = \mathcal{S}\cup\mathcal{X}\cup \mathcal{Y} = \{S, X_1, \cdots, X_N, Y\}$ with a joint distribution $P$ over $\mathcal{V}$
that factorizes over $\mathcal{G}$, where $Y$ is the target variable of interest and $S$ is the sensitive attribute. The parents of a variable $X_i$ is denoted as $pa(X_i)$ with a direct edge pointing from $X_j \in pa(X_i)$ to $X_i$. We use $do(X_i = x)$ to represent a hard intervention that assigns the value $x$ to $X_i$. This intervention will remove all the edges from $pa(X_i)$ to $X_i$; correspondingly, soft interventions modify conditional distributions between $pa(X_i)$ and $X_i$. A soft intervention chooses a conditional distribution from $P_k(X_i|pa(X_i)), k = 1,\cdots, K$ and does not change the causal model $\mathcal{G}(\mathcal{V},\mathcal{E})$; it only changes $P$ at variable $X_i$. Previous work tries to find the best interventions that maximizing the reward $Y$. Motivated by ensuring fairness in an online marketplace, we try to answer the question: how to find the best \textit{fair} soft intervention that maximizes the outcome of target variable? We discuss the online marketplace example and our fairness definition in detail below. 

Consider an online advertisement marketplace example as shown in Figure~\ref{fig:example}. This example is motivated by the case introduced in \citet{sen2017identifying} and \citet{bottou2013counterfactual}. Revenue is determined by many variables in a complex causal graph. The supplier aims to find the best prediction algorithms among $K$ options in order to predict the click through rate, which corresponds to the soft intervention $P_i(V|pa(V)), i=1,2,\cdots, K$ where we use $V$ to represent variable that will be soft intervened. The assumption is reasonable because in practice we usually have a good understanding of our deployed algorithms and their output given input. This assumption further enables us to use importance sampling for best arm identification~\citep{sen2017identifying,lattimore2016causal}. The goal of the best arm intervention is to maximize the expected revenue $\mathbb{E}_i[Y]$ given intervention $i$. However, the best intervention is not always fair, in the sense that it may prefer some majority group. For example, one algorithm can generate large revenue for male customers but little for female customers. This difference is modeled by the binary sensitive attribute $S = \{s, s'\}$. We assume the knowledge of the distribution $P(X_k|pa(X_k))$ for all $X_k \in ch(S)$, where $ch(X)$ is the set of child nodes of $X$. We further assume that $S$ has no parents in the causal graph. These two assumptions are natural since most sensitive attributes like race and gender are not determined by outside factors, and the probabilities $P(X_k|pa(X_k))$
, like user query given a male user, 
can be estimated using observational data. They also ensure the validity of our approach to use importance sampling for the counterfactual estimation. We assume $P(X_k|pa(X_k))$ is continuous with respect to each other or have same non-zero support for discrete distribution.
In addition to maximizing the revenue $\mathbb{E}_i[Y]$, we also need to make sure the algorithm deployed is fair with a high probability. It is important for online platforms to make sure that fairness is achieved when deploying their algorithms. Otherwise, minority groups might lose interest in the platform and drop out eventually, leaving the platform prone to miss future business opportunities~\citep{zhang2019group}. Our algorithm requires only \textbf{\textit{partial}} knowledge of $X_k, pa(X_k)$ for $X_k \in ch(S)$ and $S, V, pa(V)$; all other variables in the causal graph may be unobserved across the experiment.

\subsection{Fairness Definition}

Many fairness definitions are adopted in various applications. A common class of approaches for fair inference is to quantify fairness via an associative, rather than causal, relationship between the sensitive feature $S$ and the outcome $Y$. For instance, ~\citep{feldman2015certifying} adopted the
80\% rule, for comparing selection rates based on sensitive features. This is a guideline (not a legal test) advocated by the Equal Employment Opportunity Commission~\citep{eeoc1979} to suggest possible discrimination. In our problem, we ask the question whether the requester will get the same outcome if the sensitive attribute is different --- if they vary the sensitive attribute while fixing other attributes in the causal model. This is referred to as counterfactual fairness ~\citep{kusner2017counterfactual}. Counterfactual quantity concerns ``the value of $Y$ if $S$ had taken value s'', similar to the notation in ~\citet{pearl2009causality,kusner2017counterfactual}, we denote it by $Y_{S\leftarrow s}(X)$. Inspired by ~\citet{kusner2017counterfactual}, we propose a counterfactual measure $\mathbb{E}Y_i^{s',s} = \mathbb{E}[Y_{S\leftarrow s'}(X | S = s)]$ that leads to our fairness definition. This quantity is different from a simple do operation on $S$ since the evidence is still from $S=s$ but evaluated in a counterfactual world. Unlike ~\citet{kusner2017counterfactual}, we define fairness over a continuous outcome, which can be revenue, salary or time served depending on specific application. In order to estimate the counterfactual effect, recall the assumption that we have access to $P(X_k|pa(X_k))$ for all $X_k \in ch(S)$, and we can estimate the counterfactual by importance sampling. We define the counterfactual fairness with respect to the final outcome as:
$
    \zeta_i^{s',s} = \mathbb{E}Y_i^{s',s}- \mathbb{E}Y_i^{s,s} 
$.
This leads to the definition of fair interventions:

\begin{definition}[$\mathcal{E}$-fair intervention] An intervention $i$ is said to be an $\mathcal{E}$-fair intervention for attribute $S\in\{s, s'\}$ if it satisfies $|\zeta_i^{s',s}| < \mathcal{E}$ and $|\zeta_i^{s,s'}| < \mathcal{E}$. 
\end{definition}

The hyperparameter $\mathcal{E}$ is problem dependent, and different platforms and socioeconomic conditions will influence the choice of $\mathcal{E}$. 


\subsection{Budget Constraint}
The problem we consider involves a fixed budget $B$, and pulling arm $i$ incurs a cost $c_i$. After pulling each arm, we observe $\{S, X, pa(X), Y\}$ for all $X \in \{ch(S),V\}$. Additionally, we have the option to do hard intervention by recruiting volunteers with desired sensitive attributes $s \in S$. Denote the cost for obtaining a data point with $S=s$ for arm $k$ as $a_k$, $S=s'$ as $b_k$. It is required that the average cost of sampling and revealing does not exceed a cost budget $B$, i.e. $\sum_{k=0}^{K-1}c_k \nu_{Y,k} + a_k \nu_{s,k} +b_k \nu_{s',k}\leq B$, where $\nu_{Y,k}, \nu_{s,k}, \nu_{s',k}$ are the fraction of time that we pull $k$-th arm (doing k-th soft intervention) given $S=s$, $S=s'$  respectively, and $\sum_k \nu_{Y,k}+ \nu_{s,k}+ \nu_{s',k} = 1$.

\subsection{Objective}
The common objective of simple regret for best-arm identification is not well-suited in our problem because for unfair arms, the regret is undefined. Therefore, we restrict our theoretical analysis to bounding the error probability of the problem defined as $e(B) = \mathbb{P}(\hat{k}(B)\neq k^\star)$.
Here $k^\star$ is the best \emph{fair} intervention, which is defined as the arm that has the largest expect outcome among all arms that are $\mathcal{E}$-fair. However, if the cost of choosing an unfair arm is known a priori, we can easily extend our proof for simple regret. 

\begin{figure*}[htp]
    \centering
    \includegraphics[width=4in]{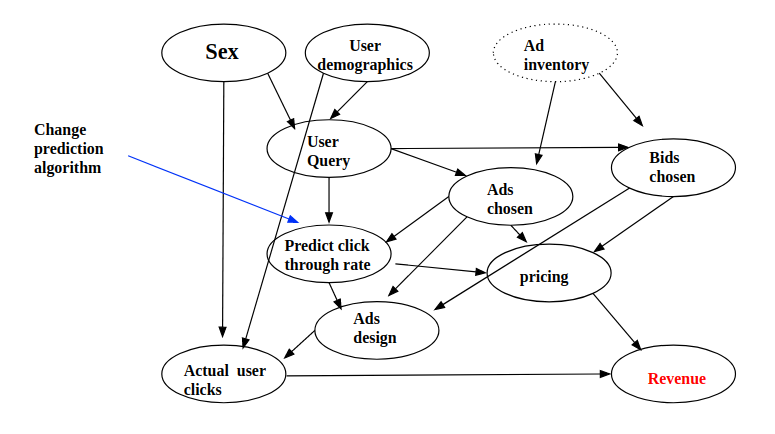}
    \caption{An example motivated by ~\citet{sen2017identifying} and ~\citet{bottou2013counterfactual} of finding the best intervention on an online advertisement marketplace, the supplier need to choose the best prediction algorithm to maximize the revenue. Here Sex is the sensitive attribute we want to protect and ad inventory is hidden variable, our algorithm only need to observe part of the graph for the constraint best arm identification problem.}
    \label{fig:example}
\end{figure*}

\section{Main Result}

In this section, we provide our main theoretical result on a successive rejection
style algorithm that leverages the information leakage and simultaneously identifies fair arms via importance sampling. Since we know the distribution of different interventions, after pulling one arm, we can use importance sampling to help estimate the reward of all other arms. Moreover, we use importance sampling clipping from \citep{lattimore2016causal,sen2017identifying} to control the variance in the estimation of outcome and fairness. We provide theoretical guarantee on the probability of error under budget constraint $B$: $e(T,B)$. 


The samples from arm $j \in [K]$ can be used to estimate $\mathbb{E}_i[Y]$, the outcome of arm $i$. Formally, 
    $\mathbb{E}_i[Y]=\mathbb{E}_j\left[Y\frac{P_i(V|pa(V))}{P_j(V|pa(V))}\right]$.
The equation is valid because the only change of the joint probability in the causal model under arm $i$ and $j$ is at $P(V|pa(V))$. Similarly, we have the estimates of the counterfactual effect
\begin{align}
    \mathbb{E}_{i,s'}[Y^{s,s'}]&=\mathbb{E}_{j,s'}\left[Y\frac{P_i(V|pa(V))}{P_j(V|pa(V))} 
   \times \frac{\prod_{X_i\in ch(S)}P(X_i|pa(X_i)\backslash S,S=s)}{\prod_{X_i\in ch(S)}P(X_i|pa(X_i)\backslash S, S=s')} \right],
\end{align}
because the only difference in the product of the joint distribution $\mathcal{P}$ is $\prod_{X_i\in ch(S)}P(X_i|pa(X_i))$. 

In the spirit of~\citep{sen2017identifying}, we use the notion of $f$-divergence to define the clip estimators. 
\begin{definition}
The conditional $f$-divergence between two distributions $p(X,Y)$ and $q(X,Y)$ is defined as 
\begin{align}
    D_f(p_{X|Y}\|q_{X|Y})=\mathbb{E}_{q_{X,Y}}\left[f\left(\frac{p_{X|Y}(X|Y)}{q_{X|Y}(X|Y)}\right)\right] 
\end{align}
\end{definition}
In the special case where the function in the $f$-divergence is 
$f_1(x) = x\exp(x-1)-1$, the cut-off threshold in the clipped estimator is selected as
$M_{ij} = 1+\log(1+D_{f_1}(P_i\|P_j))$.

Identifying a fair arm involves a similar cut-off threshold for the estimation of fairness criterion. 
\begin{align*}
    &D_{ij}^{s,s'} = 
    \log\!\Bigg(\!\mathbb{E}_{i,s}\!\!\left[\exp\left|\frac{P_i}{P_j}\!\left(\!\frac{P_s}{P_{s'}}\!-\!1\!\right)\!\right|\right] \! + \! \mathbb{E}_{i,s'}\!\!\left[\exp\left|\frac{P_i}{P_j}\!\left(\!\frac{P_s}{P_{s'}}\!-\!1\!\right)\!\right|\right]\!\!\!\Bigg).
\end{align*}
Here $P_i$ is the shorthand for $P_i(V|pa(V))$ and $P_s$ represents $
\prod_{X_i\in ch(S)}P(X_i|pa(X_i)\backslash S, S=s)$. This measure naturally arises in our theoretical analysis and helps to bound the failing probability of identifying unfair arms in each stage. Due to the unbounded variance brought by importance sampling, the $K$-by-$K$ matrices $M_{ij}$ and $D_{i,j}$ are used to control the variance of our clipped estimators for outcome and fairness measure in Equation \eqref{eqn:yest1} and \eqref{eqn:fest1}. 
\begin{align}
    \label{eqn:yest1}
    \hat{Y}_k^\epsilon \!&=\!  \frac{1}{Z_k}\!\!\!\sum_{\substack{m\in\mathcal{T}_j \\ 1 \leq j \leq K}}\!\!\!\!\frac{1}{M_{kj}} Y_j(m)\frac{P_k(V_j(m)|pa(V)_j(m))}{P_j(V_j(m)|pa(V)_j(m))} \cdot
     \mathbb{I}\!\Big\{\!\!\frac{P_k(V_j(m)|pa(V)_j(m))}{P_j(V_j(m)|pa(V)_j(m))}\!\leq\! 2\log\!\left(\!\frac{2}{\epsilon}\!\right)\!M_{kj}\!\Big\}
\end{align}

\begin{align}
    \label{eqn:fest1}
    \hat{\zeta}_k^{s,s'}(\epsilon)\! &=\!\!  \frac{1}{O_k^{s,s'}}\!\!\sum_{\substack{1\leq j \leq K \\ m\in\mathcal{T}_{j,s'}}} \!\!
    \frac{1}{D_{kj}^{s,s'}}Y_{j,s'}(m)\frac{P_k}{P_j}\left(\frac{P_s}{P_{s'}} - 1\right)\!
    \mathbb{I}\Big\{\left|\frac{P_k}{P_j}\left(\frac{P_s}{P_{s'}} - 1\right)\right|\!\leq\! 2\log\left(\frac{2}{\epsilon}\right)D_{kj}^{s,s'}\Big\}
\end{align}

We aggregate samples from different arms and use the estimators in Equation \eqref{eqn:yest1} and \eqref{eqn:fest1} for the expected outcome and fairness constraint. The estimator in Equation \eqref{eqn:fest1}, unlike the one-sided estimator in ~\citet{sen2017identifying}, is two-sided. Let $\mathcal{T}_k \subset \{1,2,\cdots, \tau\}$ be the indices of all the samples collected from arm $k$, and let $\mathcal{T}_{k,s} \subset \{1,2,\cdots, \tau\}$ be the indices of all the samples collected from arm $k$ with sensitive attribute $s$. We use $\tau_i$, $\tau_{i,s'}$ to denote the samples collected for estimating the outcome $Y$ and for estimating the fairness constraint $\zeta^{s,s'}$. We use $m$ to denote the index of samples. Finally, let $Z_k=\sum_j\frac{\tau_j}{M_{kj}}$, $O_k^{s,s'} = \sum_j\frac{\tau_{j,s'}}{D_{kj}^{s,s'}}$. The constant $\epsilon$ is phase-dependent and is chosen to be $2^{-(l-1)}$ in the proof to balance the bias-variance trade-off in each stage. 

\subsection{Algorithm}
To find the best fair intervention, we proposed a variant of success rejection style algorithm that we named constraint successive rejection (CSR).
The pseudo-code of our methods are listed in Algorithm~\ref{alg:main}, \ref{alg:mainv2} and \ref{alg:budget}. At each phase, we identify safe fair arms set $\mathcal{F}$ and jointly eliminate sub-optimal arms (arms that are fair but have lower reward than the optimal fair arm) and unfair arms and get remaining arms $\mathcal{R}$. The step size in each phase of each set is carefully designed to control the probability of mis-identification. A naive solution to this problem is to use a two-stage algorithm that identifies all fair arms in the first phase and uses any best-arm identification module in the second phase to identify the best arm. We will also experiment with this approach as a baseline in Numeric Analysis Section. Intuitively, our joint approach should have a better guarantee on the probability of error because samples are used more efficiently. We set $n(T) = \lceil \log2\log10\sqrt{T}\rceil$ and $\overline{\log}(n) = \sum_{i=1}^n\frac{1}{i}$; in Algorithm \ref{alg:budget}, $\tau(l) = \frac{T}{l\overline{\log}(n(T))}$ and $\sum_i \tau(l) = T$. The three matrices used in budget allocation algorithm are decided by the divergence matrices and are used to solve a linear programming problem. The $\mathbf{A}, \mathbf{B}, \mathbf{C}$ matrices in Algorithm \ref{alg:budget} corresponds to the $K$-by-$K$ outcome divergence matrix and two fairness divergence matrix filled with $M_{ij},D_{i,j}^{s,s'},D_{i,j}^{s',s}$. By solving the max-min problem, we can ensure a uniform good approximation in estimating the outcome and the fairness constraint. It returns $\tau_{Y,j}, \tau_{s,j}, \tau_{s',j}$, which are the number of samples allocated to estimate $\hat Y_H$, $\hat \xi^{s,s'}$ and $\hat \xi^{s',s}$. This allocation considers the budget constraint and information leakage through outcome and fairness divergence and tries to optimize the variance of the worst arm. Moreover, if after running Algorithm 1, $|\mathcal{F}|=0$, we declare that no fair arm is found. 

\begin{minipage}{0.46\textwidth}
\begin{algorithm}[H]
\caption{Constraint Successive Rejection (CSR) - Given $T, B$}
\begin{algorithmic}
\label{alg:main}
\FOR{$l = 1$ to $n(T)$}
\STATE get $\tau_Y(l), \tau_{s'}(l),\tau_s(l)$ using Algorithm \ref{alg:budget}
\STATE Use arm $k$, $\tau_Y(l), \tau_{s'}(l),\tau_s(l)$ times and collect samples $(Y,X,pa(X),S)$, for 
\FOR{$k\in\mathcal{R}$}
\STATE Calculate $\hat{Y}_k, \hat{\zeta}_k^{s,s'}, \hat{\zeta}_k^{s',s}$
\ENDFOR 
\IF{$|\mathcal{R}| = 1$}
 \STATE return arm in $\mathcal{R}$
\ENDIF
\STATE $\mathcal{F} = \Big\{k\in \mathcal{R}(l): \hat{\zeta}_k^{s,s'} + \frac{3}{2^l} < \mathcal{E} \land \hat{\zeta}_k^{s,s'} - \frac{3}{2^l} > -\mathcal{E} \land \hat{\zeta}_k^{s',s} + \frac{3}{2^l} < \mathcal{E} \land \hat{\zeta}_k^{s',s} - \frac{3}{2^l} > -\mathcal{E}\Big\} $
\IF{$|\mathcal{F}| = 0$} 
    \STATE go to next phase 
\ENDIF 
\STATE $\hat{Y}_H = \underset{k\in \mathcal{F}}{\max}\hat{Y}_k$
\STATE $\mathcal{R} = \mathcal{R} - \{k\in\mathcal{R}:\big(\hat{Y}_H > \hat{Y}_k + \frac{5}{2^l}\big)\lor 
\big(
\hat{\zeta}_k^{s,s'} - \frac{3}{2^l} > \mathcal{E}
\big)\lor 
\big(
\hat{\zeta}_k^{s,s'} + \frac{3}{2^l} < -\mathcal{E}
\big) \lor 
\big(
\hat{\zeta}_k^{s',s} - \frac{3}{2^l} > \mathcal{E} 
\big)\lor 
\big(
\hat{\zeta}_k^{s',s} + \frac{3}{2^l} < -\mathcal{E}
\big)
\}$
\ENDFOR
\end{algorithmic}
\end{algorithm}
\end{minipage}
\hspace{1em}
\begin{minipage}{0.55\textwidth}
\begin{minipage}{0.9\linewidth}
\begin{algorithm}[H]
\caption{Constraint Successive Rejection (CSR) - Given $T, B$ (V2)}
\begin{algorithmic}
\label{alg:mainv2}
\STATE Same as Algorithm \eqref{alg:main} but calculate $\hat{Y}_k$,$ \hat{\zeta}_k^{s,s'}$, $\hat{\zeta}_k^{s',s}$ from all the samples collected so far.
\end{algorithmic}
\end{algorithm}
\end{minipage}
\\
\begin{minipage}{0.9\linewidth}
\begin{algorithm}[H]
\caption{Budget Allocation}
\begin{algorithmic}
\label{alg:budget}
\STATE ALLOCATE($c, a, b, B, \mathbf{A}, \mathbf{B}, \mathbf{C}, \mathcal{R}, \tau$)
\begin{align*}
    \frac{1}{\sigma^\star(B, \mathcal{R})} & = v^\star(B,\mathcal{R}) \\ 
    & = \max_{\mathbf{v}}\min_{k\in\mathcal{R}} [\mathbf{A}\nu_{Y}, \mathbf{B}\nu_s, \mathbf{C}\nu_{s'}]_k \\ \nonumber 
    s.t. & \sum_{i=0}^K c_i\nu_{Y,i} + a\nu_{s,i} + b\nu_{s',i} \leq B\text{, } \\ \nonumber
    &\sum_{i=1}^K\nu_{Y,i}+\nu_{s,i}+\nu_{s',i} = 1 \quad \text{, and } \\ \nonumber 
    &\nu_{Y,i},\nu_{s,i},\nu_{s',i} \geq 0 \nonumber
\end{align*}
\STATE return 
\begin{align*}
        \quad \tau_{Y,j}, \tau_{s,j}, \tau_{s',j} = \nu_{Y,j}^\star\tau ,  \nu_{s,j}^\star\tau  , \nu_{s',j}^\star\tau
\end{align*}
\end{algorithmic}
\end{algorithm}
\end{minipage}
\end{minipage}
\subsection{Theoretical Result}
We provide a problem-dependent theoretical guarantee for Algorithm \ref{alg:main}. Theorem 1 shows that CSR guarantees the exponential decaying error rate for returning the optimal arm when \textit{best fair} arm exists. When there is no fair arm, Theorem 2 shows that CSR guarantees exponential decaying error rate of returning no fair arms found. 
\begin{theorem}
Consider a problem instance with $K$ candidate arms. Suppose $K>1$ and there exists a best fair arm $k^*$ satisfying $\abs{\xi_{k^*}^{s,s'}},\abs{\xi_{k^*}^{s',s}} < \mathcal{E}$. Let $\rho_k$ be the optimal gap associated with each arm defined in the proof. Further define 
    $\bar{H} = \max_{k\neq k^\star} \frac{\rho_k^3}{2^{-2\rho_k}v^\star(B,R^\star(k))^2}$.
The error probability of Algorithm 1 is bounded as 
\begin{align}
        e(T,B) 
    &  \leq  8K^2\rho^\star\exp(-\frac{T}{8\bar{H}\overline{\log}(n(T))}),
\end{align}
where $\rho^\star$ and $\bar{H}$ are problem-dependent constants. 
\end{theorem}

\textbf{Proof Sketch:}
To derive this result, note that error happens when the best fair arm is eliminated in a certain round. The elimination occurs (i) when a fair arm's estimated reward exceeds the best fair arm or (ii) when the best fair arm is deemed unfair. Both suggests a lower bound on the error of estimated reward and on the error of estimated fairness gap. The two types of error are described with our notions of ``success event'' and ``safe event'' in our online analysis. Our algorithm operates in rounds, in each stage, our algorithm tries to eliminate unfair and sub-optimal arms simultaneously. By carefully control the failing probability of each stage, we get an exponential decaying rate up to a log factor.  \qed 
\begin{remark}
In order to jointly eliminate sub-optimal and unfair arms in each phase, our algorithm has a larger constant compared to ~\citet{sen2017identifying}, where there is no constraint on fairness. The larger constant $\bar{H}$ is from more constraint on fairness for identifying fair arms, which makes our problem harder than ~\citet{sen2017identifying}. Our problem-dependent constants $\rho_k$ includes the fairness gap too, which characterize the difficulty level of identifying each arm using our algorithm. 
\end{remark}

Next, we consider the setting when no fair arm exists, the following theorem shows our algorithm will also identify this setting with a probability asymptotically converges to 1. 
  \begin{theorem}
    Suppose that $K > 1$ and for all arms $k\in [K]$ either $\abs*{\xi_k^{s,s'}} \geq \mathcal{E}$ or $\abs*{\xi_k^{s', s}} \geq \mathcal{E}$. Define the constant
    $
    \xi^* = \min_{k\in [K]} \min(\abs{\abs*{\xi_k^{s,s'}} -\mathcal{E}},\abs{\abs*{\xi_k^{s',s}}. -\mathcal{E}}).
    $
Then, the probability that the algorithm declares no fair arm exists is at least $$1 - 4 K n(T)  \exp (-\frac{(\xi^*)^2 T }{32 n(T)^3 \overline{\log}\left(n(T)\right)}v^*(B)^2),$$
 where $v^*(B) = \min_{\emptyset \neq  R \subseteq [K]} v^*(B, R)$.
  \end{theorem}
  The proof of both theorems is provided in the supplementary material due to space limit. The probability of correctly identifying the non-existence of fair arms approaches 1 as $T$ goes to infinity. 
  
\section{Numerical Analysis}
\label{sec:na}
To demonstrate the effectiveness of the proposed approach, we use the causal graph as shown in Figure~\ref{fig:my_label}. Each sensitive attribute in $S$ has a corresponding categorical distribution to generate $V \in \{1, \cdots, m\}$, and each arm $V_i$ has a binomial distribution. The random variable $\epsilon$ has a binomial distribution with parameter $p=0.99$, $Y$ is generated by $\epsilon$ and $V_i$ as follows:
\begin{equation}
  Y = \begin{cases}
    f(V_i), & \text{if $\epsilon=1$}.\\
    1-f(V_i), & \text{otherwise}.
  \end{cases}
\end{equation}


We follow the experiment setting in~\citet{sen2017identifying} and assume that there is only one easy arm that can be pulled many times. It corresponds to the setting that we have a sufficient number of samples from a deployed system but it is costly to acquire samples for new arms. We define the following budget constraint:
$
    \sum_{k\neq 1}\nu_{Y,k} + \nu_{s,k} + \nu_{s',k} < \frac{1}{\sqrt{T}}. 
$
This implies we cannot acquire more than $\sqrt{T}$ samples for all other arms for reward and fairness estimation. This setting makes it easier for us to adjust the divergence with arm 1 to see how our algorithm perform under various settings. We set $K=30$ and $m=20$ in our experiment and all results are averaged over $100$ runs. 

We use a two-stage style algorithm as a baseline for comparison. The two-stage algorithm uses half of the budget to identify all arms that satisfy the fairness constraint, and then uses the remaining budget to identify the best arm among the selected fair arms. For the best arm identification module, we use SRIS-v1 and SRIS-v2 in \citet{sen2017identifying}, which are the state-of-the-art algorithms of this problem to the best of our knowledge. The complete description and pseudo-code of the two stage algorithms are provided in the supplementary material. We abbreviate our algorithm CSR-V1, CSR-V2 to V1 and V2, respectively. Similarly, the two baseline algorithms are referred to as \emph{two stage V1} and \emph{two stage V2}. We report the probability of error in our plots below. Each setting has a unique best \emph{fair} arm that our algorithm aims to identify. We plot the probability of error against the total budget in each figure.

\tikzset{node/.style={circle,fill=gray!10,draw,minimum size=0.4cm,inner sep=0pt} }
\tikzset{arc/.style = {->,> = latex, thick, } }
\begin{figure*}
\centering
\begin{minipage}[c]{0.12\textwidth}
\centering
    \begin{tikzpicture}[auto,node distance = 0.1 cm, scale = 0.6] 
    \node[node] at(-1,2.2) (S) { S };
    \node[node] at(-1,1) (V) { V };
    \node[node] at(0, 0) (Y) { Y };
    \node[node] at(1, 1) (E) { $\epsilon$ };
    \node[above right = .1cm and .15 cm of V]  (T) { $P_1, .., P_K$ };
    \draw[arc] (S) to node{ } (V); 
    \draw[arc] (V) to node{ } (Y); 
    \draw[arc] (E) to node{ } (Y);
    \draw[gray,->>] (T) to [bend left] (V);
    \end{tikzpicture} 
    \caption{Causal graph.
    }
    \label{fig:my_label}
\end{minipage}
\hfill
\begin{minipage}[c]{0.85\textwidth}
\centering
\subfigure[No Fairness Constraint]{\label{fig:nofair}\includegraphics[width=0.32\textwidth]{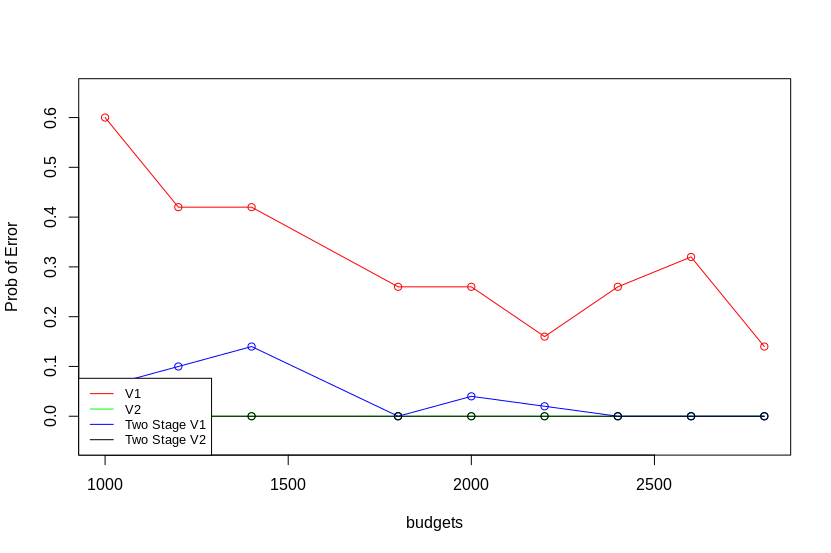}}
\subfigure[Loose Fairness Constraint]{\label{fig:largefair}\includegraphics[width=0.32\textwidth]{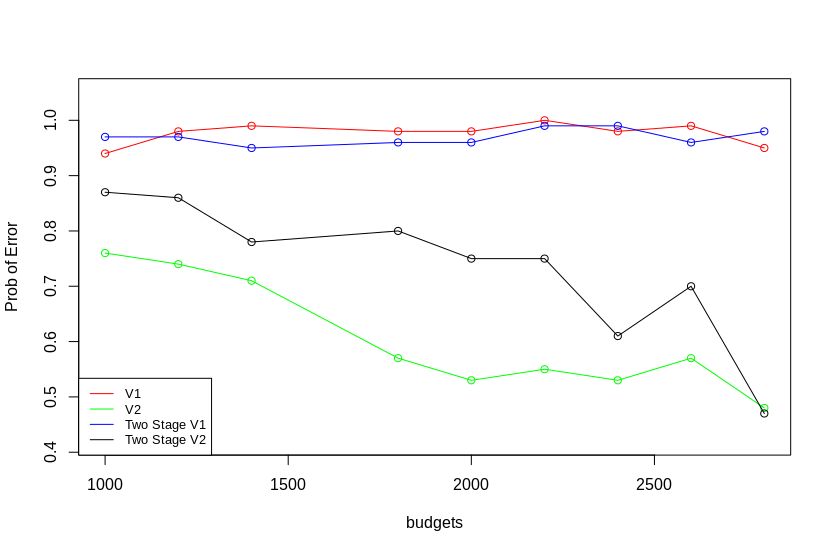}}
\subfigure[Tight Fairness Constraint]{\label{fig:smallfair}\includegraphics[width=0.32\textwidth]{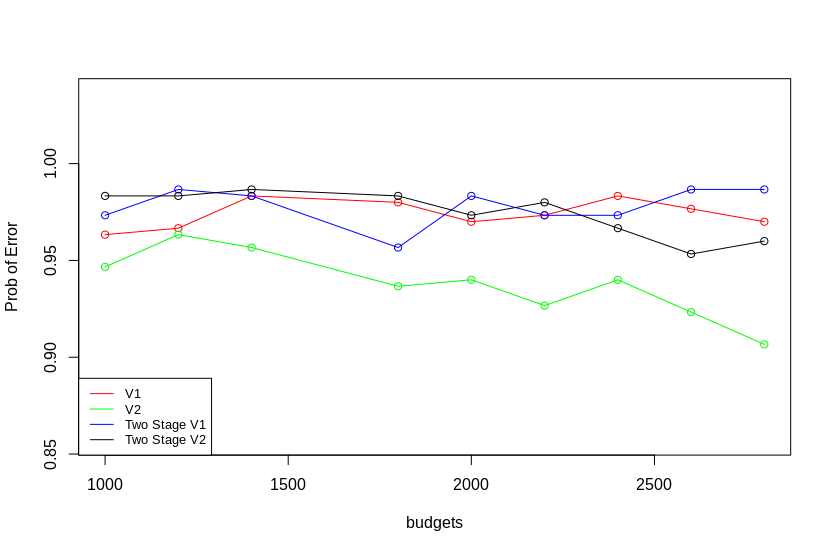}}
\caption{(Red: CSR-V1, Green: CSR-V2, Blue: Two Stage - SRISV1, Purple: Two Stage - SRISV2)}
\label{fig:mainexp}
\end{minipage}
\end{figure*}


\subsection{No Fairness Constraint}
First, we consider the case where the fairness constraint is always satisfied so that the problem degenerates to best arm identification in~\citet{sen2017identifying}. We set the divergence relatively small so that $M_{k0}/\sqrt{K}, D_{k0}^{s,s'}/\sqrt{K}, D_{k0}^{s',s}/\sqrt{K}<30$. Since we mainly use samples from arm $1$ to estimate arm $K$'s outcome and counterfactual fairness, we keep this condition in the following experiments. We set the constraint $\mathcal{E}$ equal to $2$, which all arms satisfy by a large margin. Since our algorithm will spend extra budget on exploring the fairness constraint on sensitive attributes, it should perform worse than a normal best arm identification algorithm. The observation is confirmed in Figure~\ref{fig:nofair}. V2 and Two Stage V2 easily achieve perfect performance, and two stage V1 demonstrates low probability of error because it still spends half of the budget for the best arm identification module. 


\subsection{Loose Fairness Constraint}
In this setting, we consider the case that the fairness constraint is relatively large in the sense that some arms do not meet the fairness constraint, and we need to identify the optimal arm that is fair. We keep the same constraint on the divergence and keep the minimum gap on the final outcome similar. As the fairness constraint becoming smaller, the minimal gap on fairness gets smaller and the problem becomes harder. As a result, only $8$ arms are not feasible. As shown in Figure \ref{fig:largefair},  V2 achieves the best performance in minimizing probability of error comparing to all other baselines. V1 has a similar performance comparing to its two-stage counterpart. This could be due to noisy estimation since only limited samples are used in each phase. 


\subsection{Tight Fairness Constraint}
In this setting, we consider the hardest case where the minimal gap is relatively small in the sense that many arms do not meet the fairness constraint and we need to identify the optimal arm that is fair. Since the fairness constraint is very small, the algorithms cannot identify feasible arms easily, which makes them slower to find the optimal arm. In our experiment, only $9$ arms are feasible. The fairness constraint is set to $0.03$, and the optimal gap of outcome is around $0.06$, which is also fairly small. The final result is shown in Figure \ref{fig:smallfair}. We notice the result is a little noisy when the constraint is small, so the result is averaged over $300$ runs. In this setting, all algorithms suffer from high probability of error and clearly a lot more samples are needed than previous two settings to correctly identify the best fair arm. Two stage V1 and two stage V2 have a very similar performance to V1. However, V2 shows much better performance and a clear decreasing trend when the budget is increasing. This empirically demonstrates the effectiveness of our proposed algorithm.

\subsection{Ablation Study}
Finally, we consider the effect of divergence on the result of our algorithm.  We control the divergence of $M_{k1}$ , $D_{k1}^{s,s'}$ and $D_{k1}^{s',s}$ here to investigate how our results change with them. We use the same causal graph and experimental setting as above. Our analysis suggests that larger divergence implies smaller information leakage and a harder identification process. In all experiments performance of algorithms drops while the divergence gets larger. Our ablation study empirically supports our observation that a lower divergence term facilitates the estimation by importance sampling, which in turn improves the identification of the best \emph{fair} arm. The complete results and discussions are included in supplementary material due to space limit. 

\section{Real Dataset}
In order to examine how different methods perform in practice. We also use a more challenging task with a large causal graph HEPAR-II studied in medicine~\citep{onisko1998probabilistic}. HEPAR-II has 70 nodes and 123 arcs. The average degree of the bayesian network is 3.51 and maximum in-degree  is 6. An subset of the Directed Acyclic Graph~\citep{pearl2009causality} is shown in Figure ~\ref{fig:causal-real}.

\tikzset{node/.style={ellipse,fill=gray!10,draw,minimum size=0.8cm,inner sep=0pt} }
\tikzset{arc/.style = {->,> = latex, thick, } }

\begin{figure}
\centering
    \begin{tikzpicture}[auto,node distance = 0.8 cm, scale = 1.1]
    \node[node] at(-1.2, 4.1) (1) { ... };
    \node[node] at(-0.4, 4.3) (2) { ... };
    \node[node] at(0.4, 4.1) (3) { ... };
    \node[node] at(-0.4,3.3) (ChHepatitis) { ChHepatitis };
    \node[node] at(-1.8, 1.8) (Steatosis) { ... };
    \node[] at(-0.4,2.7) (p) {};
    \node[node] at(-0.4,2) (fibrosis) { fibrosis };
    \node[node] at(-1,1) (Cirrhosis) { Cirrhosis };
    \node[node] at(0, 0) (carcinoma) { carcinoma };
    \node[node] at(1, 1) (PBC) { PBC };
    \node[node] at(1.5,2) (sex) {sex};
    \node[node] at(0.7,2) (age) {age};
    \node[above right = .1cm and .15 cm of fibrosis]  (T) { $P_1, .., P_K$ };
    \draw[arc] (1) to (ChHepatitis);
    \draw[arc] (2) to (ChHepatitis);
    \draw[arc] (3) to (ChHepatitis);
    \draw[arc] (ChHepatitis) to node{ } (fibrosis);
    \draw[arc] (Steatosis) to node{ } (Cirrhosis);
    \draw[arc] (fibrosis) to node{ } (Cirrhosis);
    \draw[arc] (sex) to node{ } (PBC);
    \draw[arc] (age) to node{ } (PBC);
    \draw[arc] (Cirrhosis) to node{ } (carcinoma);
    \draw[arc] (PBC) to node{ } (carcinoma);
    \draw[gray,->>] (T) to [bend right] (p);
    \end{tikzpicture}
    \caption{Causal graph for the HEPAR-II dataset.
    }
    \label{fig:causal-real}
\end{figure}
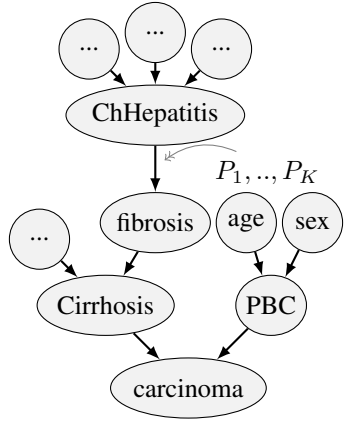

For our experiment, we choose the algorithm, protected attribute, outcome node as `fibrosis', `sex', `carcinoma' in the bayesian network respectively. We created 10 random probability distributions for `fibrosis' given its parents to represent 10 candidate algorithms that we need to identify the best fair arm from. We encode `carcinoma' as 0 and 1 to represent the final reward generated for each customer and this reward needs to be counterfactually fair with respect to sex, which means for each customer, the algorithm will ensure their reward is almost the same if they have a different gender.  We make the same assumptions as in the numeric analysis section for budget constraint, so the first arm represents an arm already deployed in the system so one can pull it for more times and other arms come with a larger cost. The fairness constraint is set to be 0.2. The \href{https://www.bnlearn.com/bnrepository/discrete-large.html#hepar2}{Bayesian Network Repository} provides a fitted bayesian network so that we can do random sampling to generate data that we consider to be ground-truth in our experiment. 

Note that our algorithm only needs partial knowledge of the causal graph, an example is the omitted nodes in Figure ~\ref{fig:causal-real}. In this experiment, we only need access to two conditional probabilities and observation of 6 nodes in the causal graph, which is only 8.6\% of the total nodes. In the context of an online marketplace, it means less data and teams are needed to ensure the algorithmic fairness for reasons like efficiency and security.

We experiment with total budget (number of pulls) ranging from 2000 to 10000 with a step size of 2000 and the probability of error is averaged over 50 runs. The result is shown in Figure ~\ref{fig:hepar}, CSR and CSR-V2 still show the best performance under HEPAR-II dataset and the conclusion is consistent with what we found in numerical analysis section. We also find SRISV2 doesn't perform well here, we suspect the first stage will eliminate the best arm by mistake in this more challenging problem. 

\begin{figure}
    \centering
    \includegraphics[width=0.2\textwidth]{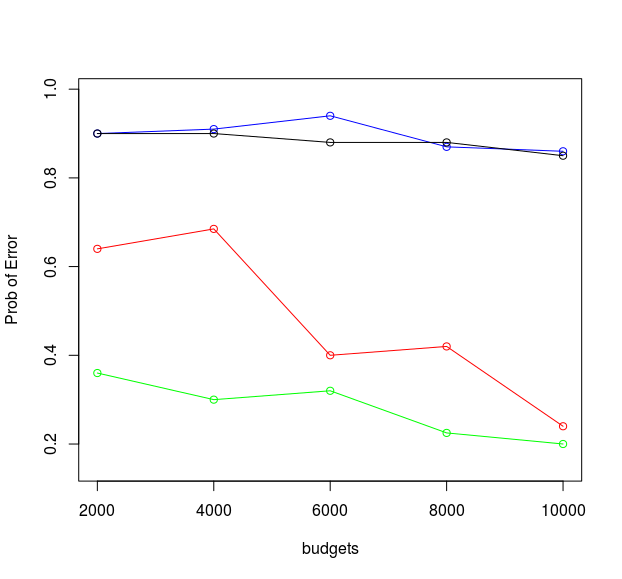}
    \caption{Result for HEPAR-II (Red: CSR-V1, Green: CSR-V2, Blue: Two Stage - SRISV1, Purple: Two Stage - SRISV2)}
    \label{fig:hepar}
\end{figure}

\section{Conclusion}
In this paper, we propose a novel fairness criterion on revenue with respect to sensitive attributes in a given causal graph. Our criterion is based on counterfactual fairness. We introduce the important problem of identifying the best arm that satisfies the fairness constraint. This problem is solved by extending an existing causal bandit algorithm to the case of constrained best arm identification. We empirically and theoretically examine the effectiveness of the proposed method. This paper considers only one binary sensitive attribute. A general setting involves many sensitive attributes that people want to protect, and they could have complex dependencies. We also refrain from the discussion of identifiability and counterfactual inference in partially or completely unknown causal graphs in this paper and assume a given causal graph with full access to desired conditionals. These more challenging, but no less important problems are left for future research.

\newpage
\bibliographystyle{plainnat}
\bibliography{example_paper} 

\begin{thebibliography}{33}
\providecommand{\natexlab}[1]{#1}
\providecommand{\url}[1]{\texttt{#1}}
\expandafter\ifx\csname urlstyle\endcsname\relax
  \providecommand{\doi}[1]{doi: #1}\else
  \providecommand{\doi}{doi: \begingroup \urlstyle{rm}\Url}\fi

\bibitem[Bolukbasi et~al.(2016)Bolukbasi, Chang, Zou, Saligrama, and
  Kalai]{bolukbasi2016man}
Tolga Bolukbasi, Kai-Wei Chang, James~Y Zou, Venkatesh Saligrama, and Adam~T
  Kalai.
\newblock Man is to computer programmer as woman is to homemaker? debiasing
  word embeddings.
\newblock In \emph{Neurlps}, pages 4349--4357, 2016.

\bibitem[Bottou et~al.(2013)Bottou, Peters, Qui{\~n}onero-Candela, Charles,
  Chickering, Portugaly, Ray, Simard, and Snelson]{bottou2013counterfactual}
L{\'e}on Bottou, Jonas Peters, Joaquin Qui{\~n}onero-Candela, Denis~X Charles,
  D~Max Chickering, Elon Portugaly, Dipankar Ray, Patrice Simard, and
  Ed~Snelson.
\newblock Counterfactual reasoning and learning systems: The example of
  computational advertising.
\newblock \emph{JMLR}, 14\penalty0 (1):\penalty0 3207--3260, 2013.

\bibitem[Calders and Verwer(2010)]{calders2010three}
Toon Calders and Sicco Verwer.
\newblock Three naive bayes approaches for discrimination-free classification.
\newblock \emph{Data Mining and Knowledge Discovery}, 21\penalty0 (2):\penalty0
  277--292, 2010.

\bibitem[Chen et~al.(2019)Chen, Cuellar, Luo, Modi, Nemlekar, and
  Nikolaidis]{chen2019fair}
Yifang Chen, Alex Cuellar, Haipeng Luo, Jignesh Modi, Heramb Nemlekar, and
  Stefanos Nikolaidis.
\newblock Fair contextual multi-armed bandits: Theory and experiments.
\newblock \emph{arXiv preprint arXiv:1912.08055}, 2019.

\bibitem[Doffman(2019)]{amz2019}
Zak Doffman.
\newblock Amazon refuses to quit selling 'flawed' and 'racially biased' facial
  recognition, 2019.
\newblock URL
  \url{https://www.forbes.com/sites/zakdoffman/2019/01/28/amazon-hits-out-at-attackers-and-claims-were-not-racist/#1c7564b946e7}.

\bibitem[Dwork et~al.(2012)Dwork, Hardt, Pitassi, Reingold, and
  Zemel]{dwork2012fairness}
Cynthia Dwork, Moritz Hardt, Toniann Pitassi, Omer Reingold, and Richard Zemel.
\newblock Fairness through awareness.
\newblock In \emph{Proceedings of the 3rd innovations in theoretical computer
  science conference}, pages 214--226, 2012.

\bibitem[EEOC.(1979)]{eeoc1979}
EEOC.
\newblock The u.s. equal employment opportunity commission uniform guidelines
  on employee selection procedures.
\newblock 1979.

\bibitem[Feldman et~al.(2015)Feldman, Friedler, Moeller, Scheidegger, and
  Venkatasubramanian]{feldman2015certifying}
Michael Feldman, Sorelle~A Friedler, John Moeller, Carlos Scheidegger, and
  Suresh Venkatasubramanian.
\newblock Certifying and removing disparate impact.
\newblock In \emph{Proceedings of the 21th ACM SIGKDD International Conference
  on Knowledge Discovery and Data Mining}, pages 259--268. ACM, 2015.

\bibitem[Gillen et~al.(2018)Gillen, Jung, Kearns, and Roth]{gillen2018online}
Stephen Gillen, Christopher Jung, Michael Kearns, and Aaron Roth.
\newblock Online learning with an unknown fairness metric.
\newblock In \emph{Neurlps}, pages 2600--2609, 2018.

\bibitem[Holstein et~al.(2019)Holstein, Wortman~Vaughan, Daum{\'e}~III, Dudik,
  and Wallach]{holstein2019improving}
Kenneth Holstein, Jennifer Wortman~Vaughan, Hal Daum{\'e}~III, Miro Dudik, and
  Hanna Wallach.
\newblock Improving fairness in machine learning systems: What do industry
  practitioners need?
\newblock In \emph{Proceedings of the 2019 CHI Conference on Human Factors in
  Computing Systems}, page 600. ACM, 2019.

\bibitem[Jabbari et~al.(2017)Jabbari, Joseph, Kearns, Morgenstern, and
  Roth]{jabbari2017fairness}
Shahin Jabbari, Matthew Joseph, Michael Kearns, Jamie Morgenstern, and Aaron
  Roth.
\newblock Fairness in reinforcement learning.
\newblock In \emph{Proceedings of the 34th ICML-Volume 70}, pages 1617--1626.
  JMLR. org, 2017.

\bibitem[Joseph et~al.(2016)Joseph, Kearns, Morgenstern, and
  Roth]{joseph2016fairness}
Matthew Joseph, Michael Kearns, Jamie~H Morgenstern, and Aaron Roth.
\newblock Fairness in learning: Classic and contextual bandits.
\newblock In \emph{Neurlps}, pages 325--333, 2016.

\bibitem[Katz-Samuels and Bagherjeiran(2019)]{katz2019feasible}
Julian Katz-Samuels and Abraham Bagherjeiran.
\newblock Feasible bidding strategies through pure-exploration bandits.
\newblock 2019.

\bibitem[Katz-Samuels and Scott(2018)]{katz2018feasible}
Julian Katz-Samuels and Clay Scott.
\newblock Feasible arm identification.
\newblock In \emph{ICML}, pages 2540--2548, 2018.

\bibitem[Katz-Samuels and Scott(2019)]{katz2019top}
Julian Katz-Samuels and Clayton Scott.
\newblock Top feasible arm identification.
\newblock In \emph{The 22nd International Conference on Artificial Intelligence
  and Statistics}, pages 1593--1601, 2019.

\bibitem[Kusner et~al.(2017)Kusner, Loftus, Russell, and
  Silva]{kusner2017counterfactual}
Matt~J Kusner, Joshua Loftus, Chris Russell, and Ricardo Silva.
\newblock Counterfactual fairness.
\newblock In \emph{Neurlps}, pages 4066--4076, 2017.

\bibitem[Lattimore et~al.(2016)Lattimore, Lattimore, and
  Reid]{lattimore2016causal}
Finnian Lattimore, Tor Lattimore, and Mark~D Reid.
\newblock Causal bandits: Learning good interventions via causal inference.
\newblock In \emph{Neurlps}, pages 1181--1189, 2016.

\bibitem[Li et~al.(2019)Li, Liu, and Ji]{li2019combinatorial}
Fengjiao Li, Jia Liu, and Bo~Ji.
\newblock Combinatorial sleeping bandits with fairness constraints.
\newblock \emph{IEEE Transactions on Network Science and Engineering}, 2019.

\bibitem[Liu et~al.(2018)Liu, He, and Shen]{liuOnTimeLastMile2018}
Sheng Liu, Long He, and Zuo-Jun~Max Shen.
\newblock On-{{Time Last Mile Delivery}}: {{Order Assignment}} with {{Travel
  Time Predictors}}.
\newblock {{SSRN Scholarly Paper}} ID 3179994, {Social Science Research
  Network}, {Rochester, NY}, May 2018.

\bibitem[Liu et~al.(2017)Liu, Radanovic, Dimitrakakis, Mandal, and
  Parkes]{liu2017calibrated}
Yang Liu, Goran Radanovic, Christos Dimitrakakis, Debmalya Mandal, and David~C
  Parkes.
\newblock Calibrated fairness in bandits.
\newblock \emph{arXiv preprint arXiv:1707.01875}, 2017.

\bibitem[Metevier et~al.(2019)Metevier, Giguere, Brockman, Kobren, Brun,
  Brunskill, and Thomas]{metevier2019offline}
Blossom Metevier, Stephen Giguere, Sarah Brockman, Ari Kobren, Yuriy Brun, Emma
  Brunskill, and Philip~S Thomas.
\newblock Offline contextual bandits with high probability fairness guarantees.
\newblock In \emph{Neurlps}, pages 14893--14904, 2019.

\bibitem[Miller(2015)]{hire2019}
Clair~C Miller.
\newblock Can an algorithm hire better than a human?
\newblock 2015.

\bibitem[Nabi and Shpitser(2018)]{nabi2018fair}
Razieh Nabi and Ilya Shpitser.
\newblock Fair inference on outcomes.
\newblock In \emph{AAAI}, 2018.

\bibitem[Nabi et~al.(2018)Nabi, Malinsky, and Shpitser]{nabi2018learning}
Razieh Nabi, Daniel Malinsky, and Ilya Shpitser.
\newblock Learning optimal fair policies.
\newblock \emph{arXiv preprint arXiv:1809.02244}, 2018.

\bibitem[Olteanu et~al.(2019)Olteanu, Castillo, Diaz, and
  Kiciman]{olteanu2019social}
Alexandra Olteanu, Carlos Castillo, Fernando Diaz, and Emre Kiciman.
\newblock Social data: Biases, methodological pitfalls, and ethical boundaries.
\newblock \emph{Frontiers in Big Data}, 2:\penalty0 13, 2019.

\bibitem[Onisko et~al.()Onisko, Druzdzel, and Wasyluk]{onisko1998probabilistic}
Agnieszka Onisko, Marek~J Druzdzel, and Hanna Wasyluk.
\newblock A probabilistic causal model for diagnosis of liver disorders.

\bibitem[Patil et~al.(2019)Patil, Ghalme, Nair, and
  Narahari]{patil2019achieving}
Vishakha Patil, Ganesh Ghalme, Vineet Nair, and Y~Narahari.
\newblock Achieving fairness in the stochastic multi-armed bandit problem.
\newblock \emph{arXiv preprint arXiv:1907.10516}, 2019.

\bibitem[Pearl(2009)]{pearl2009causality}
Judea Pearl.
\newblock \emph{Causality}.
\newblock Cambridge university press, 2009.

\bibitem[Rudin(2013)]{policing2019}
Cynthia Rudin.
\newblock Predictive policing using machine learning to detect patterns of
  crime.
\newblock 2013.

\bibitem[Sen et~al.(2017)Sen, Shanmugam, Dimakis, and
  Shakkottai]{sen2017identifying}
Rajat Sen, Karthikeyan Shanmugam, Alexandres~G Dimakis, and Sanjay Shakkottai.
\newblock Identifying best interventions through online importance sampling.
\newblock In \emph{Proceedings of the 34th ICML-Volume 70}, pages 3057--3066.
  JMLR. org, 2017.

\bibitem[Yabe et~al.(2018)Yabe, Hatano, Sumita, Ito, Kakimura, Fukunaga, and
  Kawarabayashi]{yabe2018causal}
Akihiro Yabe, Daisuke Hatano, Hanna Sumita, Shinji Ito, Naonori Kakimura,
  Takuro Fukunaga, and Ken-ichi Kawarabayashi.
\newblock Causal bandits with propagating inference.
\newblock In \emph{ICML}, pages 5508--5516, 2018.

\bibitem[Zhang and Bareinboim(2018)]{zhang2018fairness}
Junzhe Zhang and Elias Bareinboim.
\newblock Fairness in decision-making—the causal explanation formula.
\newblock In \emph{Thirty-Second AAAI Conference on Artificial Intelligence},
  2018.

\bibitem[Zhang et~al.(2019)Zhang, Khaliligarekani, Tekin,
  et~al.]{zhang2019group}
Xueru Zhang, Mohammadmahdi Khaliligarekani, Cem Tekin, et~al.
\newblock Group retention when using machine learning in sequential decision
  making: the interplay between user dynamics and fairness.
\newblock In \emph{Neurlps}, pages 15243--15252, 2019.

\end{thebibliography}
\onecolumn
\appendix

%
\section{Proofs}

\subsection{Importance Sampling Clipped Estimator}
\begin{definition}{$\eta_{i,j}(\epsilon)$}
\begin{align}
    \eta_{i,j}(\epsilon) = \min \left\{\eta: 
    \mathbb{P}_i\left(\frac{P_i(V|pa(V))}{P_j(V|pa(V))}>\eta\right) \leq \frac{\epsilon}{2}
    \right\}
\end{align}
\end{definition}

\begin{definition}{$\gamma_{i,j}^{s,s'}(\epsilon)$}
\begin{align}
    \gamma_{i,j}^{s,s'}(\epsilon) = \min  \Big\{\gamma : \; 
    &\mathbb{P}_{i,s}\left(\Big|\frac{P_i(V|pa(V))}{P_j(V|pa(V))}\left(\prod_{X_k\in ch(S)}\frac{P(X_k)|pa(X_k)\setminus S, S=s)}{P(X_k)|pa(X_k)\setminus S, S=s')}-1\right)\Big| >\gamma\right) + \\ \nonumber
    & 
    \mathbb{P}_{i,s'}\left(\Big|\frac{P_i(V|pa(V))}{P_j(V|pa(V))}\left(\prod_{X_k\in ch(S)}\frac{P(X_k)|pa(X_k)\setminus S, S=s)}{P(X_k)|pa(X_k)\setminus S, S=s')}-1\right)\Big| >\gamma\right) 
    \leq \frac{\epsilon}{2}
    \Big\}
\end{align}
\end{definition}

The clipping threshold defined above are used in the estimators below. 

\begin{align}
    \label{eqn:yest}
    \hat{Y}_i^{(\eta)}(j) = \frac{1}{t}\sum_{m=1}^t Y_j(m)\frac{P_i(V_j(m)|pa(V)_j(m))}{P_j(V_j(m)|pa(V)_j(m))}\mathbb{I}\Big\{\frac{P_i(V_j(m)|pa(V)_j(m))}{P_j(V_j(m)|pa(V)_j(m))}\leq \eta_{i,j}(\epsilon)\Big\}
\end{align}

\begin{align}
    \label{eqn:fest}
    \hat{\zeta}_{i,j}^{s,s'}(\gamma) = &\frac{1}{t}\sum_{m=1}^{t}Y_{j,s'}(m) \frac{P_i(V_{j,s'}(m)|pa(V)_{j,s'}(m))}{P_j(V_{j,s'}(m)|pa(V)_{j,s'}(m))}
    \left(\prod_{X_k\in ch(S)}\frac{P(X_{k_{j,s'}}|pa(X_k)_{j,s'}\setminus S, S=s)}{P(X_{k_{j,s'}}|pa(X_k)_{j,s'}\setminus S, S=s')} - 1\right)
    \\ \nonumber
    &\mathbb{I}\Big\{
    \frac{P_i(V_{j,s'}(m)|pa(V)_{j,s'}(m))}{P_j(V_{j,s'}(m)|pa(V)_{j,s'}(m))}
    \Big|\prod_{X\in ch(S)}\frac{P(X_{k_{j,s'}}|pa(X_k)_{j,s'}\setminus S, S=s)}{P(X_{k_{j,s'}}|pa(X_k)_{j,s'}\setminus S, S=s')} - 1\Big| \leq \gamma_{i,j}^{s,s'}
    \Big\}
\end{align}

\begin{align}
    \hat{\zeta}_{i,j}^{s',s}(\gamma) = &\frac{1}{t}\sum_{m=1}^{t}Y_{j,s}(m) \frac{P_i(V_{j,s}(m)|pa(V)_{j,s}(m))}{P_j(V_{j,s}(m)|pa(V)_{j,s}(m))}
    \left(\prod_{X_k\in ch(S)}\frac{P(X_{k_{j,s}}|pa(X_k)_{j,s}\setminus S, S=s')}{P(X_{k_{j,s}}|pa(X_k)_{j,s}\setminus S, S=s)} - 1\right)
    \\ \nonumber
    &\mathbb{I}\Big\{
    \frac{P_i(V_{j,s}(m)|pa(V)_{j,s}(m))}{P_j(V_{j,s}(m)|pa(V)_{j,s}(m))}
    \Big|\prod_{X_k\in ch(S)}\frac{P(X_{k_{j,s}}|pa(X_k)_{j,s}\setminus S, S=s')}{P(X_{k_{j,s}}|pa(X_k)_{j,s}\setminus S, S=s)} - 1\Big| \leq \gamma_{i,j}^{s',s}
    \Big\}
\end{align}

We provide guarantees for the estimators above. We use the shorthand $\mu_i = \mathbb{E}_i[Y]$, which the expected outcome under $i$-th arm. 

\begin{lemma}
$ \hat{Y}_i^{(\eta)}(j)$ and $\hat{\zeta}_{i,j}^{s,s'}(\gamma)$ satisfy the following bounds:
\begin{align}
    \label{eqn:ybound}
    \mathbb{E}_j[\hat{Y}_i^{(\eta)}(j)]\leq \mu_i \leq \mathbb{E}_j[\hat{Y}_i^{(\eta)}(j)] + \frac{\epsilon}{2}
\end{align}

\begin{align}
\label{eqn:zetabound}
    \mathbb{E}\hat{\zeta}_{i,j}^{s,s'}(\gamma) - \frac{\epsilon}{2} \leq \zeta_{i,j}^{s,s'} \leq \mathbb{E}\hat{\zeta}_{i,j}^{s,s'}(\gamma) + \frac{\epsilon}{2}
\end{align}

\begin{align}
\label{eqn:yvarbound}
\mathbb{P}\left(\mu_i-\delta-\frac{\epsilon}{2}\leq\hat{Y}_i^{(\eta)}(j)\leq\mu_i+\delta\right)\geq 1-2\exp(-\frac{\delta^2 t}{2\lambda_{i,j}(\epsilon)^2})
\end{align}

\begin{align}
\label{eqn:zetavarbound}
    \mathbb{P}\left(\zeta_{i,j}^{s,s'}-\frac{\epsilon}{2}-\delta\leq \hat{\zeta}_{i,j}^{s,s'}(\gamma) \leq \zeta_{i,j}^{s,s'}+\frac{\epsilon}{2}+\delta\right)\geq 1-2\exp(-\frac{\delta^2 t}{8\gamma_{i,j}^{s,s'}(\epsilon)^2}).
\end{align}

\end{lemma}
\begin{proof}
The proof of Equation \eqref{eqn:ybound} is exactly the same as ~\citet{sen2017identifying}, we refer readers to read supplementary of ~\citet{sen2017identifying} for details, we show proof for Equation \eqref{eqn:zetabound},\eqref{eqn:yvarbound},\eqref{eqn:zetavarbound} below. 
\begin{align}
    \nonumber
    &\zeta_{i,j}^{s,s'} =  \mathbb{E}_{j,s'}Y \frac{P_i(V|pa(V))}{P_j(V|pa(V))}\times 
    (\prod_{X_k\in ch(S)}\frac{P(X_k)|pa(X_k)\setminus S, S=s)}{P(X_k)|pa(X_k)\setminus S, S=s')}-1) \\ \nonumber
    &  
    \mathbb{I}\Big\{\Big|
    \frac{P_i(V|pa(V))}{P_j(V|pa(V))}(\prod_{X_k\in ch(S)}\frac{P(X_k)|pa(X_k)\setminus S, S=s)}{P(X_k)|pa(X_k)\setminus S, S=s')}-1)\Big| \leq \gamma_{i,j}^{s,s'}(\epsilon)\Big\}
     \\ \nonumber
    &  + \mathbb{E}_{j,s'}Y \frac{P_i(V|pa(V))}{P_j(V|pa(V))}\times (\prod_{X_k\in ch(S)}\frac{P(X_k)|pa(X_k)\setminus S, S=s)}{P(X_k)|pa(X_k)\setminus S, S=s')}-1)\\ 
    &  
    \mathbb{I}\Big\{\Big|
    \frac{P_i(V|pa(V))}{P_j(V|pa(V))}(\prod_{X_k\in ch(S)}\frac{P(X_k)|pa(X_k)\setminus S, S=s)}{P(X_k)|pa(X_k)\setminus S, S=s')}-1)\Big| > \gamma_{i,j}^{s,s'}(\epsilon) \Big\} \\ \nonumber. 
\end{align}
Denote the second term as $\omega$, we have
\begin{align}
    |\omega| \leq &\mathbb{P}_{i,s}(\Big|\frac{P_i(V|pa(V))}{P_j(V|pa(V))}(\prod_{X_k\in ch(S)}\frac{P(X_k)|pa(X_k)\setminus S, S=s)}{P(X_k)|pa(X_k)\setminus S, S=s')}-1)\Big| >\gamma_{i,j}^{s,s'}) + \\ \nonumber
    & \quad \mathbb{P}_{i,s'}(\Big|\frac{P_i(V|pa(V))}{P_j(V|pa(V))}(\prod_{X_k\in ch(S)}\frac{P(X_k)|pa(X_k)\setminus S, S=s)}{P(X_k)|pa(X_k)\setminus S, S=s')}-1)\Big| >\gamma_{i,j}^{s,s'})\leq \frac{\epsilon}{2}.
\end{align}
All terms in summation of \ref{eqn:yest} and \ref{eqn:fest} are bounded by $\eta_{i,j}(\epsilon)$ and $\gamma_{i,j}^{s,s'}(\epsilon)$ respectively, and by Azuma-Hoeffding:
\begin{align}
    \mathbb{P}(|\hat{Y}_i^{(\eta)}(j) - \mathbb{E}_j[\hat{Y}_i^{(\eta)}(j)]|>\delta)\leq 2\exp(-\frac{\delta^2 t}{2\eta_{i,j}(\epsilon)^2})
\end{align}

\begin{align}
    \mathbb{P}(\left|\hat{\zeta}_{i,j}^{s,s'}(\gamma) - \mathbb{E}_{j,s'}\hat{\zeta}_{i,j}^{s,s'}(\gamma)\right|>\delta)\leq 2\exp(-\frac{\delta^2 t}{8\gamma_{i,j}^{s,s'}(\epsilon)^2}). 
\end{align}

\end{proof}

\subsection{Relating with $f$-divergence}
Note that
\begin{align}
    \mathbb{E}_i\left[\exp(\frac{P_i(V|pa(V))}{P_j(V|pa(V))})\right] = [1+D_{f_1}(P_i||P_j)]e, 
\end{align}
where $f_1(x) = x\exp(x-1) - 1$

\begin{lemma}
\begin{align}
\label{eqn:Mdivproof}
    \eta_{i,j}(\epsilon)\leq
    \log(\frac{2}{\epsilon}) + 1+\log(D_{f_1}(P_i||P_j))\leq
    2\log(\frac{2}{\epsilon})[1+\log(D_{f_1}(P_i||P_j))] = 2\log(\frac{2}{\epsilon})M_{ij}
\end{align}

\begin{align}
    \label{eqn:Ddivproof}
    \gamma_{i,j}^{s,s'}(\epsilon) \leq 2\log(\frac{2}{\epsilon}) D_{i,j}^{s,s'}
\end{align}
\end{lemma}

\begin{proof}

The similar proof for \eqref{eqn:Mdivproof} is done in ~\citet{sen2017identifying}, here we extend the proof for Equation \eqref{eqn:Ddivproof}.

With moment generating function, we have the bound
\begin{align}
    &\mathbb{P}_{i,s}(\Big|\frac{P_i(V|pa(V))}{P_j(V|pa(V))}(\prod_{X_k\in ch(S)}\frac{P(X_k)|pa(X_k)\setminus S, S=s)}{P(X_k)|pa(X_k)\setminus S, S=s')}-1)\Big| >\gamma) +  \\ \nonumber 
    & \mathbb{P}_{i,s'}(\Big|\frac{P_i(V|pa(V))}{P_j(V|pa(V))}(\prod_{X_k\in ch(S)}\frac{P(X_k)|pa(X_k)\setminus S, S=s)}{P(X_k)|pa(X_k)\setminus S, S=s')}-1)\Big| >\gamma) \\ \nonumber
    &\quad \quad \leq \exp(-\gamma)\Bigg(\mathbb{E}_{i,s}[\exp|\frac{P_i}{P_j}(\frac{P_s}{P_{s'}}-1)|] + \mathbb{E}_{i,s'}[\exp|\frac{P_i}{P_j}(\frac{P_s}{P_{s'}}-1)|]\Bigg). 
\end{align}
Suppose the right hand side is smaller than $\frac{\epsilon}{2}$, we have 
\begin{align}
    \gamma \geq \log(\frac{2}{\epsilon}) + \log\Bigg(\mathbb{E}_{i,s}[\exp|\frac{P_i}{P_j}(\frac{P_s}{P_{s'}}-1)|] + \mathbb{E}_{i,s'}[\exp|\frac{P_i}{P_j}(\frac{P_s}{P_{s'}}-1)|]\Bigg)
\end{align}

By the definition of $\gamma_{i,j}^{s,s'}(\epsilon)$, we have:
\begin{align}
    \gamma_{i,j}^{s,s'}(\epsilon) &\leq \log(\frac{2}{\epsilon}) + \log\Bigg(\mathbb{E}_{i,s}[\exp|\frac{P_i}{P_j}(\frac{P_s}{P_{s'}}-1)|] + \mathbb{E}_{i,s'}[\exp|\frac{P_i}{P_j}(\frac{P_s}{P_{s'}}-1)|]\Bigg) \\ \nonumber
    & \leq 
    2\log(\frac{2}{\epsilon}) D_{i,j}^{s,s'} 
\end{align}
where the last inequality follows from $a + b \leq 2 ab$ for $a \geq 1$ and $b \geq \log_2 2$.
\end{proof}

Lemma~A.1 and Lemma~A.2 directly implies the following theorem. 

\begin{theorem}
\begin{align}
    \mathbb{P}(\mu_i-\delta-\frac{\epsilon}{2}\leq \hat{Y}_i^{(\eta)}(j) \leq \mu_i + \delta) \geq 1-2\exp(-\frac{\delta^2 t}{8\log(\frac{2}{\epsilon})^2{M_{i,j}}^2})
\end{align}

\begin{align}
    \mathbb{P}(\zeta_{i,j}^{s,s'}-\delta-\frac{\epsilon}{2}\leq\hat{\zeta}_{i,j}^{s,s'}(\gamma) \leq \zeta_{i,j}^{s,s'} +\delta+\frac{\epsilon}{2}) \geq 1 - 2\exp(-\frac{\delta^2 t}{32\log(\frac{2}{\epsilon})^2 {D_{i,j}^{s,s'}}^2})
\end{align}
\end{theorem}

\subsection{Aggregating Estimators from Different Arms}
Let $\mathcal{T}_k \subset \{1,2,\cdots, \tau\}$ be the indices of all the samples collected from arm $k$, $\mathcal{T}_{k,s} \subset \{1,2,\cdots, \tau\}$ be the indices of all the samples collected from arm $k$ with sensitive attribute $s$. Recall the estimators $\hat{Y}_k^\epsilon$ and $\hat{\zeta}_k^{s,s'}(\epsilon)$ defined in the main paper:

\begin{align}
    \hat{Y}_k^\epsilon = \frac{1}{Z_k}\sum_{j=0}^K\sum_{m\in\mathcal{T}_j}\frac{1}{M_{kj}}Y_j(m)\frac{P_k(V_j(m)|pa(V)_j(m))}{P_j(V_j(m)|pa(V)_j(m))}\mathbb{I}\Big\{\frac{P_k(V_j(m)|pa(V)_j(m))}{P_j(V_j(m)|pa(V)_j(m))}\leq 2\log(\frac{2}{\epsilon})M_{kj}\Big\}
\end{align}

\begin{align}
    \hat{\zeta}_k^{s,s'}(\epsilon) = \frac{1}{O_k^{s,s'}}\sum_{j=0}^K\sum_{m\in\mathcal{T}_{j,s'}}\frac{1}{D_{kj}^{s,s'}}Y_{j,s'}(m)\frac{P_k}{P_j}(\frac{P_s}{P_{s'}} - 1)\mathbb{I}\Big\{|\frac{P_k}{P_j}(\frac{P_s}{P_{s'}} - 1)|\leq 2\log(\frac{2}{\epsilon})D_{kj}^{s,s'}\Big\}
\end{align}

where $Z_k=\sum_j\frac{\tau_j}{M_{kj}}$, $O_k^{s,s'} = \sum_j\frac{\tau_{j,s'}}{D_{kj}^{s,s'}}$
\begin{lemma}
\begin{align}
    \mathbb{E}\hat{Y}_k^\epsilon\leq \mu_k \leq \mathbb{E}\hat{Y}_k^\epsilon + \frac{\epsilon}{2}
\end{align}

\begin{align}
    \mathbb{E}\hat{\zeta}_{i,j}^{s,s'}(\epsilon)-\frac{\epsilon}{2}\leq \zeta_{i,j}^{s,s'} \leq \mathbb{E}\hat{\zeta}_{i,j}^{s,s'}(\epsilon) + \frac{\epsilon}{2}
\end{align}
\end{lemma}

\begin{proof}
With Lemma~A.2, the estimators (24) and (25) satisfy the assumption of Lemma A.1, whose proof carries over. 
\end{proof}

\begin{theorem}
\begin{align}
    \mathbb{P}(\mu_k-\delta-\frac{\epsilon}{2}\leq \hat{Y}_k^\epsilon\leq \mu_k+\delta) \geq 1-2\exp(-\frac{\delta^2\tau}{8\log(\frac{2}{\epsilon})^2}(\frac{Z_k}{\tau})^2)
\end{align}

\begin{align}
    \mathbb{P}(\zeta_{k}^{s,s'}-\delta-\frac{\epsilon}{2}\leq\hat{\zeta}_{k}^{s,s'}(\epsilon) \leq \zeta_{k}^{s,s'} +\delta+\frac{\epsilon}{2}) \geq 1 - 2\exp(-\frac{\delta^2 \tau_{s'}}{32\log(\frac{2}{\epsilon})^2}(\frac{O_k^{s,s'}}{\tau_{s'}})^2)
\end{align}
\end{theorem}
\begin{proof}
In the summation, each term is bounded by $2\log(\frac{2}{\epsilon})$, let $\bar{Y}_k^\epsilon = \frac{Z_k}{\tau}\hat{Y}_k^\epsilon$, $\bar{\zeta}_{k}^{s,s'}(\epsilon) = \frac{O_k^{s,s'}}{\tau_{s'}}\hat{\zeta}_{k}^{s,s'}(\epsilon)$, by Chernoff's bound, 
\begin{align}
    &\mathbb{P}(|\bar{\zeta}_{k}^{s,s'}(\epsilon)-\bar{\zeta}_k| > \delta) \leq 2\exp(-\frac{\delta^2\tau_{s'}}{32\log(\frac{2}{\epsilon})^2}) \\ \nonumber
    & \rightarrow \mathbb{P}(|\hat{\zeta}_{k}^{s,s'}(\epsilon)-\hat{\zeta}_k| > \delta) \leq 2\exp(-\frac{\delta^2\tau_{s'}}{32\log(\frac{2}{\epsilon})^2}(\frac{O_k^{s,s'}}{\tau_{s'}})^2)
\end{align}
Similarly, we get the concentration bound for $\hat{Y}_k$.
\end{proof}

\subsection{Budget Allocation}
\begin{lemma}
The allocation in Algorithm ensures $\frac{Z_k}{\tau}, \frac{O_k^{s,s'}}{\tau}, \frac{O_k^{s',s}}{\tau} \geq \frac{1}{\sigma^\star(B,R)} = v^\star(B,R)$, for all $k$
\end{lemma}

\subsection{Online Analysis}
Since the infeasible arms' regret is undefined, simple regret is not meaningful in our case. We will only bound the probability of error $e(T,B)$. We define some quantities for the analysis:

\begin{itemize}
    \item $\mathcal{R}(l):$ Set of arms remaining after phase $l-1$
    \item $\hat{\zeta}_k^{s,s'}(l), \hat{\zeta}_k^{s,s'}(l):$ Fairness estimator for arm k at the end of phase $l$
    \item $\hat{Y}_k(l):$ Estimator of reward of arm $k$ at the end of phase $l$
    \item $\mathcal{F}(l)\subset\mathcal{R}(l):$ Arms that meet fairness constraint with high probabilities at the end of phase $l$, given by 
    \begin{align}
        \Big\{k\in \mathcal{R}(l): \hat{\zeta}_k^{s,s'} + \frac{3}{2^l} < \mathcal{E} \land \hat{\zeta}_k^{s,s'} - \frac{3}{2^l} > -\mathcal{E} \land \hat{\zeta}_k^{s',s} + \frac{3}{2^l} < \mathcal{E} \land \hat{\zeta}_k^{s',s} - \frac{3}{2^l} > -\mathcal{E}\Big\}
    \end{align}
    \item $\hat{Y}_H(l):$ Highest estimator of reward in $\mathcal{F}(l)$ at the end of phase $l$
    \item $\mathcal{A}(l)\subset \mathcal{R}(l):$  
    \begin{align}
        \Big\{
        k\in\mathcal{R}(l): 
        \Big(Y_{k^\star} - Y_k \geq \frac{10}{2^l}\Big) \lor \Big(\zeta_k^{s,s'}-\frac{6}{2^l}>\mathcal{E} \lor 
        \zeta_k^{s,s'}+\frac{6}{2^l}<-\mathcal{E}
        \Big) \lor
        \Big(\zeta_k^{s',s}-\frac{6}{2^l}>\mathcal{E} \lor 
        \zeta_k^{s',s}+\frac{6}{2^l}<-\mathcal{E}
        \Big)
        \Big\}
    \end{align}
    \item $S_l:$ Success event of phase $l$ defined as:
    \begin{align}
        \Bigg\{\underset{k\in\mathcal{R}(l)}{\bigcap}\Big\{
        \mu_k - \frac{3}{2^{l}} \leq 
        \hat{Y}_k(l)\leq \mu_k + \frac{1}{2^{l-1}}
        \Big\} \cap \{k^\star \in \mathcal{F}(l)\} \Bigg\}
        \bigcap \\
        \qquad \Bigg\{
        \bigcap_{k\in\mathcal{R}(l), k\neq k^\star}
        \Big\{
        |\hat{\zeta}_k^{s,s'} - \zeta_k^{s,s'}| < \frac{3}{2^l}
        \cap 
        |\hat{\zeta}_k^{s',s} - \zeta_k^{s',s}| < \frac{3}{2^l}
        \Big\} \Bigg\}
    \end{align}
    \item $E_l:$ Safe event of phase $l$ defined as:
    \begin{align}
        \Bigg\{\underset{k\in\mathcal{R}(l)}{\bigcap}\Big\{
        \mu_k - \frac{3}{2^{l}} \leq 
        \hat{Y}_k(l)\leq \mu_k + \frac{1}{2^{l-1}}
        \Big\} 
        \bigcap
        \Bigg\{
        \bigcap_{k\in\mathcal{R}(l), k\neq k^\star}
        \Big\{
        |\hat{\zeta}_k^{s,s'} - \zeta_k^{s,s'}| < \frac{3}{2^l}
        \cap 
        |\hat{\zeta}_k^{s',s} - \zeta_k^{s',s}| < \frac{3}{2^l}
        \Big\} \Bigg\}
    \end{align}
\end{itemize}

If $S_l$ happens, arms in $\mathcal{A}_l$ is eliminated and optimal arm survived.
\begin{align}
    \hat{Y}_H \geq \hat{Y}_k \geq \mu_k-\frac{3}{2^{l}}; \quad
    \hat{Y}_k(l) \leq \mu_k + \frac{1}{2^{l-1}}
\end{align}
This implies $\hat{Y}_H -\hat{Y}_k \geq \frac{5}{2^l}$

\begin{align}
    \zeta_k^{s,s'} + \frac{6}{2^l} > \mathcal{-E};\quad\zeta_k^{s,s'} - \frac{6}{2^l} > \mathcal{E}; \quad |\hat{\zeta}_k^{s,s'}- \zeta_k^{s,s'}| < \frac{3}{2^l}
\end{align}
implies $\hat{\zeta}_k^{s,s'} - \frac{3}{2^l} > \mathcal{E}$, $\hat{\zeta}_k^{s,s'} + \frac{3}{2^l} < \mathcal{E}$, by symmetry the bounds with $s$ and $s'$ exchanged, so the proof also holds. 

If $E_{l}$ happens, the optimal arm is survived, given $E_l$,

\begin{align}
    |\hat{\zeta}_k^{s,s'}- \zeta_k^{s,s'}| < \frac{3}{2^l}; \quad -\mathcal{E} < \zeta_k^{s,s'} < \mathcal{E};
\end{align}

therefore $-\mathcal{E} - \frac{3}{2^l}< \hat{\zeta}_k^{s,s'} < \mathcal{E} + \frac{3}{2^l}$, also we know if $k\in \mathcal{F}(l)$, then $k\in\mathcal{F}$, so $Y_k \leq Y_{k^\star}$,  
\begin{align}
    \hat{Y}_H(l) \leq \mu_{k^\star} + \frac{1}{2^{l-1}}; \quad \hat{Y}_{k^\star}(l) \geq \mu_{k^\star} - \frac{3}{2^{l}}
\end{align}
which implies $\hat{Y}_H(l) - \hat{Y}_{k^\star}(l) \leq \frac{5}{2^l}$, so given $E_l$, optimal arm is safe from elimination.

Let $B_l = \mathbb{P}(S_l^c|S_{1:l-1}), C_l = \mathbb{P}(E_l^c|E_{1:l-1})$
\begin{lemma}
\begin{align}
    C_l = \mathbb{P}(E_l^c|E_{1:l-1}) \leq 6|\mathcal{R}^\star(l)|\exp(-\frac{2^{-2(l-1)}\tau(l)v^\star(B, R(l))^2}{32l^2})
\end{align}

and when $l > \max\{\log\frac{5}{\zeta_{k^\star}^{s,s'}+\mathcal{E}},\log\frac{5}{-\zeta_{k^\star}^{s,s'}+\mathcal{E}} \}$, 
\begin{align}
    B_l = \mathbb{P}(S_l^c|S_{1:l-1}) \leq 8|\mathcal{R}^\star(l)|\exp(-\frac{2^{-2(l-1)}\tau(l)v^\star(B, R(l))^2}{32l^2})
\end{align}
\end{lemma}

\begin{proof}
\begin{align}
    &\mathbb{P}(\mu_k-\frac{3}{2^l}\leq \hat{Y} \leq \mu_k +\frac{1}{2^{l-1}})
    \mathbb{P}(\zeta_k^{s,s'}-\frac{3}{2^l}\leq \hat{\zeta}_k^{s,s'} \leq \zeta_k^{s,s'} +\frac{3}{2^{l}}) 
    \mathbb{P}(\zeta_k^{s',s}-\frac{3}{2^l}\leq \hat{\zeta}_k^{s',s} \leq \zeta_k^{s',s} +\frac{3}{2^{l}})
    \\ \nonumber
    &\geq (1-2\exp(-\frac{2^{-2(l-1)}\tau_Y(l)}{8l^2}(\frac{Z_k}{\tau_Y(l)})^2))
    (1-2\exp(-\frac{2^{-2(l-1)}\tau_{s'}(l)}{32l^2}(\frac{O_k^{s,s'}}{\tau_{s'}(l)})^2))
    (1-2\exp(-\frac{2^{-2(l-1)}\tau_s(l)}{32l^2}(\frac{O_k^{s',s}}{\tau_{s}(l)})^2)) \\ \nonumber
    &\geq (1-2\exp(-\frac{2^{-2(l-1)}Z_k^2}{8l^2\tau(l)})
    (1-2\exp(-\frac{2^{-2(l-1)}{O_k^{s,s'}}^2}{32l^2\tau(l)})
    (1-2\exp(-\frac{2^{-2(l-1)}{O_k^{s',s}}^2}{32l^2\tau(l)}) \\ \nonumber
    &\geq (1-2\exp(-\frac{2^{-2(l-1)}\tau(l)}{8l^2}(\frac{Z_k}{\tau(l)})^2))
    (1-2\exp(-\frac{2^{-2(l-1)}\tau(l)}{32l^2}(\frac{O_k^{s,s'}}{\tau(l)})^2))
    (1-2\exp(-\frac{2^{-2(l-1)}\tau(l)}{32l^2}(\frac{O_k^{s',s}}{\tau(l)})^2)) \\ \nonumber    
    &\geq 1 - 6\exp(-\frac{2^{-2(l-1)}\tau(l)v^\star(B, R(l))^2}{32l^2}) \\ \nonumber
\end{align}

Assuming $-\mathcal{E} + \frac{5}{2^l} < \zeta_{k^\star}^{s,s'}<\mathcal{E}-\frac{5}{2^l}$, then $l > \max\{\log\frac{5}{\zeta_{k^\star}^{s,s'}+\mathcal{E}},\log\frac{5}{-\zeta_{k^\star}^{s,s'}+\mathcal{E}} \} = l_0$
\begin{align}
    \mathbb{P}(k^\star\in \mathcal{F}) &= \mathbb{P}(\hat{\zeta}_{k^\star}^{s,s'} + \frac{3}{2^l} < \mathcal{E} \land \hat{\zeta}_{k^\star}^{s,s'} -\frac{3}{2^l} > -\mathcal{E}) \\ \nonumber
    &\geq \mathbb{P}(|\hat{\zeta}_{k^\star}^{s,s'}-\zeta_{k^\star}^{s,s'}| \leq \frac{1}{2^l} + \min\{\mathcal{E} - \frac{4}{2^l}-\zeta_{k^\star}^{s,s'}, \zeta_{k^\star}^{s,s'} + \mathcal{E}-\frac{4}{2^l} \}) \\ \nonumber
     &\geq 1-2\exp(-\frac{2^{-2l}\tau(l)v^\star(B,R(l))^2}{32l^2})
\end{align}
\end{proof}

\begin{align}
\mathbb{P}(S_{1:l}) \geq 1-\sum_{s=1}^{l_0}C_s - \sum_{s=l_0+1}^{l}B_s 
\end{align}

\begin{theorem}
Consider a problem instance with $K$ candidate arms, $\rho_k$ is the optimal gap associated with each arm defined in the proof. Further define
\begin{align}
    \bar{H} = \max_{k\neq k^\star} \frac{\rho_k^3}{2^{-2\rho_k}v^\star(B,R^\star(k))^2}
\end{align}
The error probablity of Algorithm 1 is bounded as 
\begin{align}
        e(T,B) 
    &  \leq  8K^2\rho^\star\exp(-\frac{T}{8\bar{H}\overline{\log}(n(T))})
\end{align}
\end{theorem}

\begin{proof}

The phase an arm is ideally deleted is $\rho_k = \min\{\max(SO_k,l_0), F_k^{s,s'}, F_k^{s',s}\}$
where 
\begin{align}
    SO_k:= l \quad\text{if}\quad \frac{10}{2^l} < \Delta_k < \frac{10}{2^{l-1}} \quad \text{else} \quad\infty
\end{align}

\begin{align}
        F_k^{s,s'}:= l \quad\text{if}\quad \frac{6}{2^{l-1}} > \zeta_k^{s,s'} -\mathcal{E} > \frac{6}{2^{l}}  \lor   \frac{6}{2^{l-1}} > -\mathcal{E}-\zeta_k^{s,s'} > \frac{6}{2^{l}} \quad \text{else}\quad\infty
\end{align}

therefore, 
\begin{align}
    &l_k \geq \log\frac{10}{\Delta_k} \quad &\text{if $k$ is Sub-Optimal Arm} \\ \nonumber 
    &l_k\geq \log\frac{6}{\min\{|\zeta_k^{s,s'} -\mathcal{E}|, |\zeta_k^{s,s'} +\mathcal{E}|\}} \quad &\text{if $k$ is $s'$-Unfair Arm} \\ \nonumber 
    &l_k\geq \log\frac{6}{\min\{|\zeta_k^{s',s} -\mathcal{E}|, |\zeta_k^{s',s} +\mathcal{E}|\}} \quad &\text{if $k$ is $s$-Unfair Arm} \\ \nonumber 
\end{align}

Let $\Delta = \underset{k\in \text{SO}}{\min}\Delta_k$, $\Delta^{s'} = \underset{k\in \text{s'-Unfair}}{\min}\min\{|\zeta_k^{s,s'} -\mathcal{E}|, |\zeta_k^{s,s'} +\mathcal{E}|\}$,
$\Delta^{s} = \underset{k\in \text{s-Unfair}}{\min}\min\{|\zeta_k^{s',s} -\mathcal{E}|, |\zeta_k^{s',s} +\mathcal{E}|\}$,
\begin{align}
    \rho^\star = \max\{\log\frac{20}{\Delta}, \log\frac{12}{\Delta^{s'}}, \log\frac{12}{\Delta^{s}}\}
\end{align}

Define 
\begin{align}
    \mathcal{R}^\star(k) = \Big\{
    a:\rho_s \geq \rho_k
    \Big\}
\end{align}
We assume all arms ideally left after $T$, $l_0= \max\{\log\frac{5}{\zeta_{k^\star}^{s,s'}+\mathcal{E}},\log\frac{5}{-\zeta_{k^\star}^{s,s'}+\mathcal{E}},\log\frac{5}{\zeta_{k^\star}^{s',s}+\mathcal{E}},\log\frac{5}{-\zeta_{k^\star}^{s',s}+\mathcal{E}} \}$ and we have $\min\{\Delta,\Delta_{s'},\Delta_s\}>\frac{10}{\sqrt{T}}$
\begin{align}
    e(T,B) &\leq 1 - \mathbb{P}(S_{1:\rho^{\star}}) \\ \nonumber
    &\leq \sum_{l=1}^{l_0}C_l + \sum_{l=l_0+1}^{\gamma^\star}B_l \\ \nonumber
    &\leq \sum_{l=1}^{\gamma^\star}8|\mathcal{R}^\star(l)|\exp(-\frac{2^{-2l}Tv^\star(B,R)^2}{8l^3\overline{\log}(n(T))})
\end{align}

\begin{align}
    e(T,B) \leq \sum_{i=1}^m8|\mathcal{R}^\star(l_i)|(l_{i+1}-l_i)\exp(-\frac{2^{-2l_i}Tv^\star(B,R)^2}{8l_i^3\overline{\log}(n(T))})
\end{align}

\begin{align}
    e(T,B) &\leq \sum_{k\neq k^\star S}8|\mathcal{R}^\star(k)|\rho^\star\exp(-\frac{2^{-2\rho_k}Tv^\star(B,R^\star(k))^2}{8\rho_k^3\overline{\log}(n(T))}) \\ \nonumber 
    &  \leq  8K^2\rho^\star\exp(-\frac{T}{8\bar{H}\overline{\log}(n(T))})
\end{align}
by the definition of $\bar{H}$ and $|\mathcal{R}^\star(k)| \leq K$.

\end{proof}

\subsection{Proof of Theorem 4.2}
  \begin{proof}[Proof of Theorem 4.2]
  When $K>1$, the algorithm will not return an arm if $\mathcal{F}$ is empty for every round. Denote the set $\mathcal{F}$ in the $l$-th round by $\mathcal{F}(l)$, we have the following sequence of bounds. 
  \begin{align}
   \mathbb{P} &(\text{declaring no fair arm}) \geq  \mathbb{P}(\mathcal{F}(l)=\emptyset\text{ for all }l \in \{1, 2, \ldots, n(T)\}) \\
  & \geq   \mathbb{P}(\forall l \in \{1, 2, \ldots, n(T)\}\text{ and } k \in [K], \exists u_k, v_k \in \{s, s'\}, \hat \xi_k^{u_k, v_k} \notin (-\mathcal{E} + \frac3{2^l}, \mathcal{E} - \frac3{2^l}))\\
  &  = \prod_{l=1}^{n(T)} \left[1-  \mathbb{P}\left(\exists k \in [K] \text{ s.t. } \hat \xi_k^{s, s'}, \hat \xi_k^{s', s}\in (-\mathcal{E} + \frac3{2^l}, \mathcal{E} - \frac3{2^l})\right)\right] \\
  & \geq \prod_{l=1}^{n(T)} \left[1- \sum_{k=1}^K \mathbb{P}\left(\hat \xi_k^{s, s'},  \xi_k^{s', s}\in (-\mathcal{E} + \frac3{2^l}, \mathcal{E} - \frac3{2^l})\right)\right] \\
  & \geq  \prod_{l=1}^{n(T)} \left[1- \sum_{k=1}^K \mathbb{P}\left(\abs{\hat \xi_k^{s, s'}-\xi_k^{s,s'}}\geq \frac3{2^l} + \xi^*,  \abs{\hat\xi_k^{s', s} - \xi_k^{s', s}} \geq \frac3{2^l} + \xi^*\right)\right] \\
  & \geq   \prod_{l=1}^{n(T)} \left[1- \sum_{k=1}^K \mathbb{P}\left(\abs{\hat \xi_k^{s, s'}-\xi_k^{s,s'}}\geq \frac3{2^l} + \xi^*,  \abs{\hat\xi_k^{s', s} - \xi_k^{s', s}} \geq \frac3{2^l} + \xi^*\right)\right] \\
  & \geq  \prod_{l=1}^{n(T)} \left[1- \sum_{k=1}^K 2 \exp (-\frac{(2/2^l+\xi^*)^2\tau_{s'}(l)}{32l^2} (\frac{O^{s, s'}_k}{\tau_{s'}(l)})^2) \cdot 2\exp (-\frac{(2/2^l+\xi^*)^2\tau_s(l)}{32l^2} (\frac{O^{s',s}_k}{\tau_{s}(l)})^2 \right]\\
  & \geq  \prod_{l=1}^{n(T)} \left[1- 4K \exp (-\frac{(2/2^l+\xi^*)^2\tau(l)}{32l^2}v^*(B, R(l))^2)\right]\\
  & \geq 1 - 4K \sum_{l=1}^{n(T)} \exp (-\frac{(\xi^*)^2 T }{32l^3 \overline{\log}\left(n(T)\right)}v^*(B, R(l))^2) \\
  & \geq 1 - 4 K n(T)  \exp (-\frac{(\xi^*)^2 T }{32 n(T)^3 \overline{\log}\left(n(T)\right)}v^*(B)^2).
  \end{align}
  In the above bounds, we first enforce that no arm is eliminated in every round and then apply the independence of samples and the union bounds. Finally, we use the concentration bounds with $\epsilon = 2^{-(l-1)}$ and $\delta = 2/2^l + \xi^*$.
  \end{proof}

\section{Two Stage Algorithm}
This section formally states the two stage algorithm that is used as a baseline in the numerical experiment. The pseudo-code is shown in algorithm~\ref{alg:main}, ~\ref{alg:mainv2}, ~\ref{alg:budget}. The two stage baseline spends half the budget for fair arm identification and the remaining budget to run the best arm identification module SRIS in ~\citet{sen2017identifying}. 
\begin{algorithm*}[!htb]
\caption{Two Stage Constraint Successive Rejection (TS-CSR) - Given $T, B$}
\begin{algorithmic}
\label{alg:main}
\STATE \# First Stage
\FOR{$l = 1$ to $n(T/2)$}
\STATE get $ \tau_{s'}(l),\tau_s(l)$ using Algorithm \ref{alg:budget} and $\nu_{Y,i}=0$
\STATE Use arm $k$, $ \tau_{s'}(l),\tau_s(l)$ times and collect samples $(Y,V,pa(V),S)$
\FOR{$k\in\mathcal{R}$}
\STATE Calculate $ \hat{\zeta}_k^{s,s'}, \hat{\zeta}_k^{s',s}$
\ENDFOR 
\IF{$|\mathcal{R}| = 1$}
 \STATE break 
\ENDIF
\STATE $\mathcal{R} = \mathcal{R} - \{k\in\mathcal{R}:
\big(
\hat{\zeta}_k^{s,s'} - \frac{3}{2^l} > \mathcal{E}
\big)\lor 
\big(
\hat{\zeta}_k^{s,s'} + \frac{3}{2^l} < -\mathcal{E}
\big) \lor 
\big(
\hat{\zeta}_k^{s',s} - \frac{3}{2^l} > \mathcal{E} 
\big)\lor 
\big(
\hat{\zeta}_k^{s',s} + \frac{3}{2^l} < -\mathcal{E}
\big)
\}$
\ENDFOR

\STATE \# Second Stage Use SRIS in \citet{sen2017identifying}
\FOR{$l = 1$ to $n(T/2)$}
\STATE get $\tau_Y(l)$ using Algorithm \ref{alg:budget} and $\nu_{s,i} =0$, $\nu_{s',i}=0$
\STATE Use arm $k$, $\tau_Y(l)$ times and collect samples $(Y,V,pa(V),S)$
\FOR{$k\in\mathcal{R}$}
\STATE Calculate $\hat{Y}_k$
\ENDFOR 
\IF{$|\mathcal{R}| = 1$}
 \STATE return arm in $\mathcal{R}$
\ENDIF
\STATE $\hat{Y}_H = \underset{k\in \mathcal{F}}{\max}\hat{Y}_k$
\STATE $\mathcal{R} = \mathcal{R} - \{k\in\mathcal{R}:\big(\hat{Y}_H > \hat{Y}_k + \frac{5}{2^l}\big)
\}$
\ENDFOR

\end{algorithmic}
\end{algorithm*}

\begin{algorithm*}[!htb]
\caption{Two Stage Constraint Successive Rejection (TS-CSR) - Given $T, B$ (V2)}
\begin{algorithmic}
\label{alg:mainv2}
\STATE Same as Algorithm \ref{alg:main} but calculate $\hat{Y}_k, \hat{\zeta}_k^{s,s'}, \hat{\zeta}_k^{s',s}$ from all the samples collected so far.
\end{algorithmic}
\end{algorithm*}

\begin{algorithm*}[!htb]
\caption{Budget Allocation}
\begin{algorithmic}
\label{alg:budget}
\STATE ALLOCATE($c, a, b, B, \mathbf{A}, \mathbf{B}, \mathbf{C}, \mathcal{R}, \tau$)
\begin{align}
    \frac{1}{\sigma^\star(B, \mathcal{R})} & = v^\star(B,\mathcal{R}) = \max_{\mathbf{v}}\min_{k\in\mathcal{R}} \min [\mathbf{A}\nu_{Y}, \mathbf{B}\nu_s, \mathbf{C}\nu_{s'}]_k \\ \nonumber 
    & s.t. \quad \sum_{i=0}^K c_i\nu_{Y,i} + a\nu_{s,i} + b\nu_{s',i} \leq B\text{, } \quad \sum_{i=1}^K\nu_{Y,i}+\nu_{s,i}+\nu_{s',i} = 1 \quad \text{, and } \quad \nu_{Y,i},\nu_{s,i},\nu_{s',i} \geq 0 \nonumber
\end{align}
\STATE return $\tau_{Y,j}, \tau_{s,j}, \tau_{s',j} = \nu_{Y,j}^\star\tau ,  \nu_{s,j}^\star\tau  , \nu_{s',j}^\star\tau$
\end{algorithmic}
\end{algorithm*}

\section{Ablation Study}

\subsection{Ablation Study}

Finally, we consider the effect of divergence on the result of our algorithm. Our analysis suggests that larger divergence implies smaller information leakage and a harder identification process. We control the divergence of $M_{k1}$ , $D_{k1}^{s,s'}$ and $D_{k1}^{s',s}$ here to investigate how our results change with them. We use the same causal graph and the experimental setting as above. All results are averaged over $100$ runs and we keep minimal gap on both outcome and fairness around $0.02$. This is a hard setting. 

\begin{figure*}[!htb]
\centering     
\subfigure[High $M$ and High $D$]{\label{fig:highhigh}\includegraphics[width=0.3\textwidth]{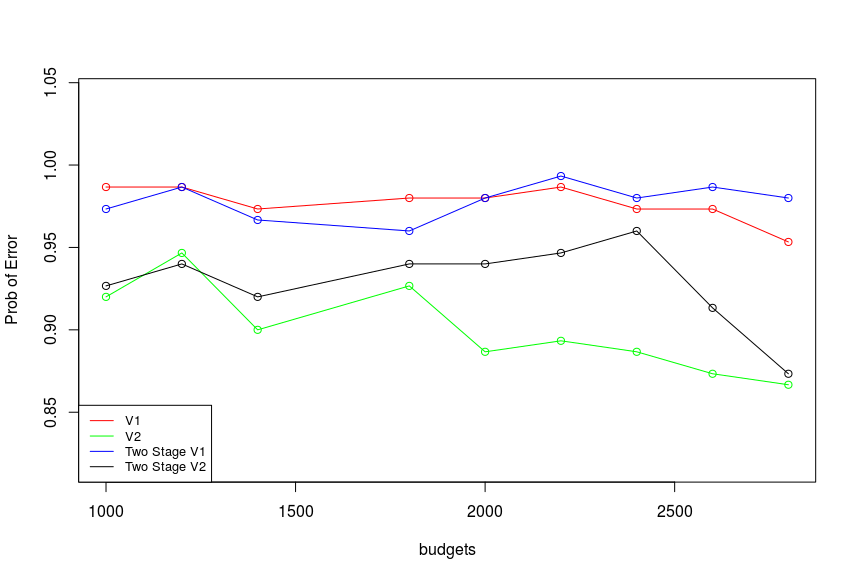}}
\subfigure[Low $M$ and High $D$]{\label{fig:lowhigh}\includegraphics[width=0.3\textwidth]{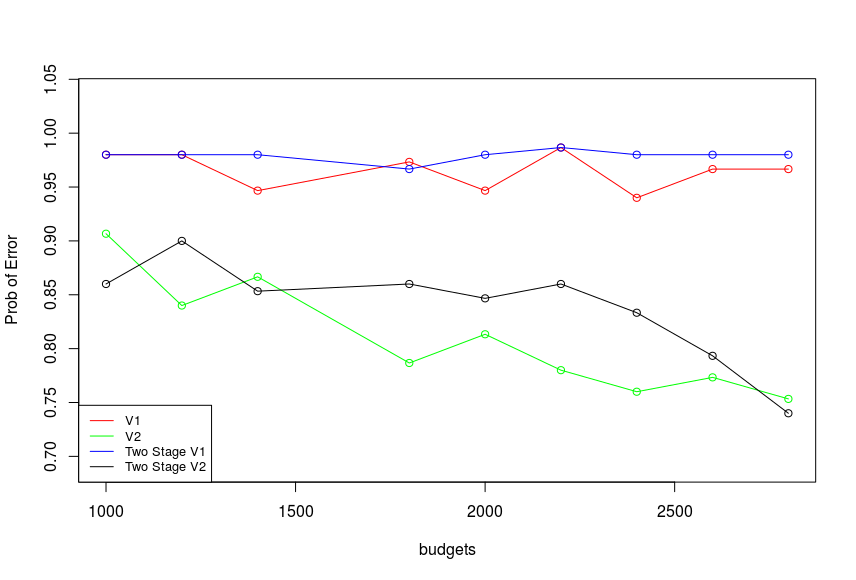}}
\subfigure[ High $M$ and Low $D$]{\label{fig:highlow}\includegraphics[width=0.3\textwidth]{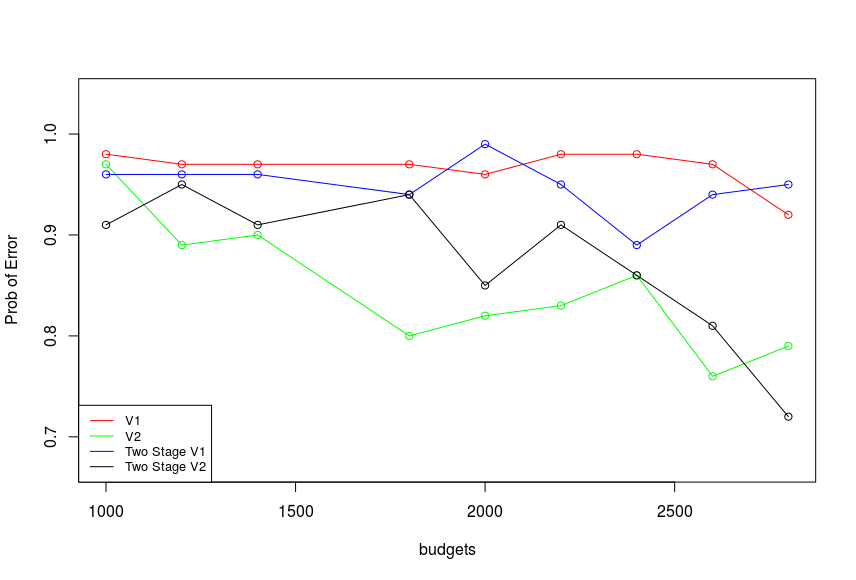}}
\caption{(Red: CSR-V1, Green: CSR-V2, Blue: Two Stage - SRISV1, Purple: Two Stage - SRISV2)}
\label{fig:ablation}
\end{figure*}

    Figure \ref{fig:highhigh} reports the case of high $M$ and high $D$, where $50>M_{i,j}>10$ and $50>D_{i,j}>10$. All algorithms have a high probability of error. Our algorithm V2 beats all other methods and its error decreases to around $0.85$. 
    
    
    Figure \ref{fig:lowhigh} reports the result for low $M$ and high $D$, where $10>M_{i,j}>0$ and $50>D_{i,j}>10$. As the divergence term of fairness becomes smaller, the information leakage improves the estimation of all arms' counterfactual fairness. The final error is around 0.75. The figure supports our analysis that smaller divergence in fairness reduces the difficulty of the problem.  
    
    
    Similarly, for high $M$ and low $D$, where $50>M_{i,j}>10$ and $10>D_{i,j}>0$, the small divergence term suggest that all algorithms should become more effective. We also empirically observe the probability of error drops to around 0.75, which shows a similar effect of $D$ comparing to $M$. 
    

Our ablation study empirically supports our observation that a lower divergence term facilitates the estimation by importance sampling, which in turn improves the identification of the best \emph{fair} arm. 

\section{Causal Graph of HEPAR-II}

\tikzset{node/.style={ellipse,fill=gray!10,draw,minimum size=0.8cm,inner sep=0pt} }
\tikzset{arc/.style = {->,> = latex, thick, } }

\begin{figure}[H]
\centering
    \resizebox{0.3\linewidth}{!}{
    \begin{tikzpicture}[auto,node distance = 0.8 cm, scale = 1.1]
    \node[node] at(-1.2, 4.1) (1) { ... };
    \node[node] at(-0.4, 4.3) (2) { ... };
    \node[node] at(0.4, 4.1) (3) { ... };
    \node[node] at(-0.4,3.3) (ChHepatitis) { ChHepatitis };
    \node[node] at(-1.8, 1.8) (Steatosis) { ... };
    \node[] at(-0.4,2.7) (p) {};
    \node[node] at(-0.4,2) (fibrosis) { fibrosis };
    \node[node] at(-1,1) (Cirrhosis) { Cirrhosis };
    \node[node] at(0, 0) (carcinoma) { carcinoma };
    \node[node] at(1, 1) (PBC) { PBC };
    \node[node] at(1.5,2) (sex) {sex};
    \node[node] at(0.7,2) (age) {age};
    \node[above right = .1cm and .15 cm of fibrosis]  (T) { $P_1, .., P_K$ };
    \draw[arc] (1) to (ChHepatitis);
    \draw[arc] (2) to (ChHepatitis);
    \draw[arc] (3) to (ChHepatitis);
    \draw[arc] (ChHepatitis) to node{ } (fibrosis);
    \draw[arc] (Steatosis) to node{ } (Cirrhosis);
    \draw[arc] (fibrosis) to node{ } (Cirrhosis);
    \draw[arc] (sex) to node{ } (PBC);
    \draw[arc] (age) to node{ } (PBC);
    \draw[arc] (Cirrhosis) to node{ } (carcinoma);
    \draw[arc] (PBC) to node{ } (carcinoma);
    \draw[gray,->>] (T) to [bend right] (p);
    \end{tikzpicture}}
    \caption{Causal graph for the HEPAR-II dataset.
    }
    \label{fig:causal-real}
\end{figure}
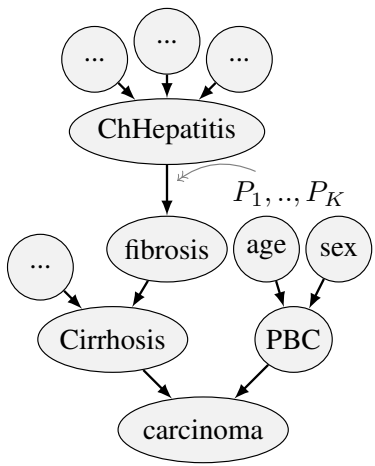

\end{document}